\documentclass[pra,twocolumn]{revtex4-1}%
\usepackage{amsfonts}
\usepackage[colorlinks=false,linkcolor=blue,citecolor=red,plainpages=false,pdfpagelabels]{hyperref}
\usepackage{amsmath}
\usepackage{amssymb}
\usepackage{graphicx}%
\usepackage{xcolor}

\usepackage{amssymb}
\usepackage{latexsym}
\usepackage{mathtools}
\providecommand{\U}[1]{\protect\rule{.1in}{.1in}}
\newtheorem{theorem}{Theorem}

\newtheorem{proposition}[theorem]{Proposition}

\newenvironment{proof}[1][Proof]{\noindent\textbf{#1.} }{\ \rule{0.5em}{0.5em}}

\newcommand{\ket}[1]{| #1 \rangle}
\newcommand{\bra}[1]{\langle #1 |}

\def\U{\mathrm{U}}

\def\B{\mathcal{B}}

\def\D{\mathcal{D}}

\def\G{\mathcal{G}}
\def\H{\mathcal{H}}
\def\U{\mathcal{U}}
\def\I{\mathcal{I}}
\def\L{\mathcal{L}}
\def\M{\mathcal{M}}
\def\N{\mathcal{N}}

\def\T{\mathcal{T}}

\def\S{\mathcal{S}}
\def\K{\mathcal{K}}

\newcommand{\tr}{\operatorname{Tr}}

\allowdisplaybreaks

\begin{document}
\title{Characterizing the performance of continuous-variable Gaussian quantum gates}
\author{Kunal Sharma}
\affiliation{Hearne Institute for Theoretical Physics, Department of Physics and Astronomy, and Center for Computation and Technology,
Louisiana State University, Baton Rouge, Louisiana 70803, USA}
\author{Mark M. Wilde}
\affiliation{Hearne Institute for Theoretical Physics, Department of Physics and Astronomy, and Center for Computation and Technology,
Louisiana State University, Baton Rouge, Louisiana 70803, USA}

\begin{abstract}
The required set of operations for universal continuous-variable quantum computation can be divided into two primary categories: Gaussian and non-Gaussian operations. Furthermore, any Gaussian operation can be decomposed as a sequence of phase-space displacements and symplectic transformations. Although Gaussian operations are ubiquitous in quantum optics, their experimental realizations are generally approximations of the ideal Gaussian unitaries. In this work, we study different performance criteria to analyze how well these experimental approximations simulate the ideal Gaussian unitaries. In particular, we find that none of these experimental approximations converge uniformly to the ideal Gaussian unitaries. However, convergence occurs in the strong sense, or if the discrimination strategy is energy bounded, then the convergence is uniform in the Shirokov--Winter energy-constrained diamond norm and we give explicit bounds in this latter case. We indicate how these energy-constrained bounds can be used for experimental implementations of these Gaussian unitaries in order to achieve any desired accuracy. 
\end{abstract}

\volumeyear{ }
\volumenumber{ }
\issuenumber{ }
\eid{ }
\date{\today}
\startpage{1}
\endpage{10}
\maketitle

\section{Introduction}
Quantum computers use quantum properties such as superposition of quantum states and entanglement for information processing and computational tasks \cite{NC10}. One of the notions of universal quantum computation consists of the manipulation of qubits encoded in discrete quantum systems and the application of a universal set of quantum operations on these qubits \cite{NC10}. 
Another way to implement discrete-variable (DV) quantum computation is to encode a finite amount of quantum information into a continuous-variable (CV) system \cite{CY95,KLM01,GKP01}. This approach is appealing given that already existing advanced optical technologies can be used for state preparation, manipulation of states, and measurement for the required quantum computational tasks \cite{LOQC07}. 

The notion of quantum computation can be further extended to CV systems, such that the transformations involved are arbitrary polynomial functions of continuous variables \cite{SB99}.  Recently, there have been many interesting advances in the context of CV quantum cryptography \cite{GG02}, CV quantum computing \cite{CVIQC17, ARW17}, and quantum machine learning \cite{LPSW17}.
One of the advantages of CV quantum computation could be in simulating CV systems more efficiently in comparison to a DV quantum computer \cite{MPSW15}.
Moreover, a hybrid of DV and CV quantum computation could be efficient for distributed quantum computing and other related tasks~\cite{SL03, PVL11, FL11}. 

The required operations for universal CV quantum computation can be divided into two primary categories: Gaussian and non-Gaussian  operations \cite{SB99, BL05}. Gaussian operations correspond to the evolution of the state of light under a Hamiltonian that is an arbitrary second-order polynomial in the electromagnetic field operators. In particular, any second-order Hamiltonian can be decomposed as a sequence of phase-space displacements (elements of the Heisenberg--Weyl group) and symplectic transformations (see, e.g., \cite{AS17} for a review). In general, along with Gaussian unitary operations, access to a Hamiltonian of at least the third power in the quadrature operators is sufficient to approximate any non-Gaussian Hamiltonian that is polynomial in the quadrature operators \cite{SB99, SL11}.

These CV Gaussian quantum gates have been extensively investigated both theoretically and experimentally in the context of quantum optics and quantum information processing \cite{P96, FMA05, MIS07, MISUM08}. In general, these quantum gates are not experimentally realized in their ideal form.
Rather, one approximates these operations using a sequence of other basic operations. For example, a displacement unitary on an arbitrary input state is commonly approximated by sending it through a particular beamsplitter along with a highly excited coherent state \cite{P96}. Moreover, squeezing and SUM transformations are generally implemented using strongly pumped nonlinear processes, which are inherently noisy, and their high sensitivity to the coupling of optical fields in a nonlinear medium makes their implementation on an arbitrary quantum state challenging \cite{FMA05}. Rather, one can approximately realize these latter gates by using a sequence of passive transformations, homodyne measurements, and off-line squeezed vacuum states \cite{FMA05, MIS07, MISUM08}.

Even if the different components involved in approximating a CV quantum gate are considered ideal, it is natural to ask the following question: in what sense does a sequence of these approximations converge to the desired quantum gate? More formally, let $\{\M^k\}_k$ denote a sequence of quantum channels corresponding to the approximations of a quantum channel $\N$. Then in what sense does the sequence $\{\M^k\}_k$ converge to  $\N$? Since these quantum channels are superoperators (completely positive and trace-preserving maps from density operators to density operators), one needs to consider various topologies on the set of quantum channels in order to study convergence. In the present context, we focus on three different notions of convergence for quantum channels: uniform and strong convergence (as presented in \cite{SH07}), and uniform convergence on the set of density operators whose marginals on the channel input have bounded energy (as presented in \cite{S18,W17}; see the appendices for more details). 

In this work, we study the aforementioned three different performance criteria to analyze how well experimental approximations simulate ideal Gaussian operations. We mainly focus on particular Gaussian unitaries, such as displacement operators, phase rotations, beamsplitters, single-mode squeezing operators, and the SUM operation, which are sufficient to generate any arbitrary Gaussian unitary operation acting on $n$ modes of the electromagnetic field \cite{BSBN02}. Results of a similar spirit, but for different examples, appeared recently in~\cite{W17,M18,BD18}.

In particular, we prove that none of these experimental approximations converge uniformly  to the ideal Gaussian processes. Qualitatively, the uniform convergence of a sequence of experimental approximations to an ideal Gaussian operation implies that the convergence is independent of the input state \cite{SH07}. As stated in \cite{SH07}, it is the same as convergence in the well-known diamond norm \cite{K97}, which is typically considered in the context of finite-dimensional quantum channels.
Therefore, our results indicate that the notion of uniform convergence for these experimental approximations of the desired Gaussian unitary operation is too strong, and we note here that similar observations have been made in the context of infinite-dimensional channels in \cite{SH07, W17,S18,LSW18} (see the appendices for more details).

Next, we study the strong convergence of these experimental approximations to the ideal Gaussian unitaries. The notion of strong convergence \cite{SH07} corresponds to the convergence of a sequence of approximations to an ideal process, considered for each possible fixed input quantum state. In particular, we show that these experimental approximations of an ideal displacement operator, beamsplitter, phase rotation, single-mode squeezer, and SUM gate converge to the ideal unitaries in the strong sense. 

A physical meaning for these two kinds of convergence was discussed in \cite{M18} by using game-theoretic arguments. In particular, it was shown that the success probability in distinguishing some CV quantum channels from their teleportation simulations is related to these two kinds of convergence, for a specific construction of the ``CV teleportation game'' \cite[Section~III]{M18}. 

One can infer from the definitions of strong and uniform convergence that the notion of strong convergence is a weaker notion of convergence,  in fact implied by uniform convergence. Another notion of convergence, which is experimentally relevant, is uniform convergence on the set of density operators whose marginals on the channel input have bounded energy (as presented in \cite{S18,W17}). Recently, it has been shown that the strong convergence of a sequence of infinite-dimensional channels is equivalent to uniform convergence on the set of energy-bounded density operators \cite{S18}. Therefore, our results imply that these experimental approximations of an ideal displacement operator, single-mode squeezer, and SUM gate converge uniformly to the ideal unitaries on the set of energy-bounded density operators. In this work, we take the energy observable to be the number operator, and we use the terminology ``energy'' and ``mean photon number'' interchangeably.

In order to experimentally approximate these different unitary operations, it is important to study how the uniform convergence over the set of energy-bounded operators depends on different experimental parameters. In particular, we consider the energy-constrained sine distance \cite[Section~12]{SWAT} as a metric to bound the Shirokov--Winter energy-constrained diamond distance between an ideal displacement operator and its experimental approximation. We first show that the fidelity between the ideal displacement and its experimental approximation when acting on a fixed input state is equal to the fidelity between a pure-loss channel and an ideal channel when acting on the same input state. We then provide an analytical expression to upper bound the Shirokov--Winter energy-constrained diamond distance between an ideal displacement and its experimental approximations, by using the recent result of \cite{N18}. Furthermore, we study different performance metrics to analyze how well an experimental approximation simulates a tensor product of different displacement operators. 

We also establish two different lower bounds on the Shirokov--Winter energy-constrained diamond distance \cite{S18, W17} between an ideal displacement operator and its experimental approximation by employing two different techniques. A first technique is based on the trace distance between the outputs of these two channels for a particular choice of the input state. In particular, we provide an analytical expression for a  lower bound on the Shirokov--Winter energy-constrained diamond distance for low values of the energy constraint.   A second technique is to estimate the Shirokov--Winter energy-constrained diamond distance by using a semidefinite program (SDP) on a truncated Hilbert space. In particular, we use an SDP from \cite{W17}, which directly follows from an SDP from \cite{JW09,JW13} defined in the context of finite-dimensional quantum channels. Moreover, we analytically show that for a fixed value of the energy constraint and for a sufficiently high value of the truncation parameter, the Shirokov--Winter energy-constrained diamond distance between two quantum channels can be estimated with an arbitrarily high accuracy by using an SDP on a truncated Hilbert space.    

Similarly, we establish analytical bounds on the Shirokov--Winter energy-constrained diamond distance between ideal beamsplitters, phase rotations, and their respective experimental approximations. 
We also study uniform convergence over the energy-bounded quantum states of some experimental approximations of both an ideal single-mode squeezing operation and a SUM gate, by considering several experimentally relevant input quantum states. 

These Gaussian unitaries are key elements for CV quantum computation \cite{SB99}, CV quantum error correction \cite{SLB98, SLB1998, LS98}, CV quantum teleportation \cite{BK98}, improving the sensitivity of an interferometer in the context of quantum metrological tasks \cite{C81}, for generating a quantum non-demolition interaction between different modes \cite{FMA05}, and for bosonic codes  \cite{chuang1997bosonic,PhysRevA592631,PhysRevLett95100501,PhysRevA75042316,niset2008experimentally,leghtas2013hardware,PhysRevA94012311,PhysRevX6031006,PhysRevA97032346,li2019designing,PhysRevA99032344,noh2019encoding,puri2019bias,albert2019pair,guillaud2019repetition,pantaleoni2019modular,wang2019quantum,grimsmo2019quantum,noh2019fault,tzitrin2019towards,eaton2019non}. Therefore, our results quantifying the performance of their experimental approximations play a critical role in understanding how to achieve any desired accuracy for several practical applications. 

The rest of the paper is organized as follows. We first briefly summarize different notions of convergence considered in this paper. We next describe experimental implementations of a displacement operator, a beamsplitter, a phase rotation, a single-mode squeezer, and a SUM gate, and then we study different notions of convergence for these gates individually. Finally, we conclude with a brief summary. We note that the appendices provide detailed proofs of all statements that follow. In the arXiv posting of this paper, we have provided all source files (Mathematica and Matlab) needed to generate the plots given in our paper, some of which rely on QETLAB \cite{qetlab} and CVX \cite{cvx}.

\section{Notions of convergence for quantum channels}

In this section, we briefly summarize three different notions of convergence for quantum channels: uniform and strong convergence (as presented in \cite{SH07}), and uniform convergence on the set of density operators whose marginals on the channel input have bounded energy (as presented in \cite{S18, W17}).

We begin by reviewing some definitions relevant for the rest of the paper (see the appendices for more details).  Let $\H$ denote an infinite-dimensional, separable Hilbert space. Let $\T(\H)$ denote the set of trace-class operators, i.e., all operators $M$ with finite trace norm: $\left\Vert M \right\Vert_1 \equiv \tr(\sqrt{M^{\dagger}M})< \infty$. Let $\D(\H)$ denote the set of density operators (positive semi-definite with unit trace) acting on $\H$. The trace distance between two quantum states $\rho, \sigma \in \D(\H)$ is given by $\left\Vert \rho - \sigma \right\Vert_1$. The fidelity between $\rho$ and $\sigma$ is defined as \cite{U76}
\begin{equation}
F(\rho, \sigma) \equiv \left\Vert \sqrt{\rho} \sqrt{\sigma} \right\Vert_1^2 .
\end{equation}
  The sine distance or $C$-distance between two quantum states $\rho$ and $\sigma$ is defined as \cite{rastegin2002relative,rastegin2003lower,rastegin2006sine,gilchrist2005distance}
\begin{equation}
  C(\rho, \sigma) \equiv \sqrt{1-F(\rho, \sigma)},
\end{equation}
and it is a metric \cite{rastegin2002relative,rastegin2003lower,rastegin2006sine,gilchrist2005distance}.  The following bounds hold between the fidelity and the trace distance between two quantum states $\rho, \sigma \in \D(\H)$:
\begin{align}\label{eq:powers}
1- \sqrt{F(\rho, \sigma)} \leq \frac{1}{2} \left\Vert \rho - \sigma \right\Vert_1 \leq \sqrt{1-F(\rho, \sigma)}~,
\end{align}
with the lower bound following from the Powers-St{\o}rmer inequality \cite{powers1970} and the upper bound from Uhlmann's theorem \cite{U76}. See also \cite{FG98}.

Let $G$ denote a positive semidefinite operator. We assume that it has discrete spectrum and that it is bounded from below. In particular, let $\{\ket{e_k}\}_k$ be an orthonormal basis for a Hilbert space $\H$, and let $\{g_k\}_k$ be a sequence of non-negative real numbers. Then 
\begin{align}\label{eq:energy-obs}
G = \sum_{k = 0}^{\infty} g_k \ket{e_k}\bra{e_k}
\end{align} 
is a self-adjoint operator that we call an energy observable.

The number operator is defined as 
\begin{align}\label{eq:num-opt}
\hat{n} = \sum_{n=0}^{\infty} n \ket{n}\bra{n},
\end{align}
where $\ket{n}$ denotes a photon-number state with $n$ photons. From \eqref{eq:energy-obs}--\eqref{eq:num-opt}, it is evident that $\hat{n}$ is an energy observable. In particular, the expectation value of $\hat{n}$ corresponds to the mean number of photons in a single-mode quantum state. Moreover, we consider the following $l$th extension $\bar{\hat{n}}_l$ of the number operator $\hat{n}$:
\begin{align}
\bar{\hat{n}}_l = \hat{n}\otimes I\otimes \dots \otimes I  + \dots + I \otimes \dots \otimes \hat{n}~,
\end{align}
where $l$ is the number of factors in each tensor product above. The expectation value of $\bar{\hat{n}}_l$ corresponds to the mean number of photons in a multi-mode quantum state. 

In our paper, we employ the two-mode squeezed vacuum state with parameter $N\geq 0$, which is defined as
\begin{equation}\label{eq:tms}
\ket{\psi_{\operatorname{TMS}}(N)} \equiv \frac{1}{\sqrt{N+1}}\sum_{n=0}^{\infty}\sqrt{\left(\frac{N}{N+1}\right)^n} \ket{n}_R\ket{n}_A~,
\end{equation}
where $\ket{n}$ again denotes a photon-number state with $n$ photons. 

It is important to note that even though the state  in \eqref{eq:tms} is a well-defined quantum state for all $N\in [0,\infty)$, the limiting object, often called ``ideal EPR state''  $\lim_{\bar{n}\to \infty}\ket{\psi_{\text{TMS}}(N)}$ \cite{EPR35}, is not a quantum state, as it is unnormalizable and it is thus not contained in the set of density operators. Similarly, the eigenvectors of the position- and momentum-quadrature operators, denoted as $\ket{x}$ and $\ket{p}$, respectively, are also not quantum states. In spite of this, the notions of uniform and strong convergence involve a supremum over the set of density operators, and so these objects can be approached in a suitable limit. Note that this point has been clarified previously in the context of uniform and strong convergence \cite{M18}.

We now recall the notion of uniform convergence for quantum channels. 
Let $\{\M^k_{A\to B}\}_k$ denote a sequence of quantum channels, where each channel takes a trace class operator acting on a separable Hilbert space  $\H_A$ to a trace class operator acting on a separable Hilbert space $\H_B$. Then the channel sequence $\{\M^k_{A\to B}\}_k$ converges uniformly to another quantum channel $\N_{A\to B}$ if the following holds:
\begin{align}\label{def:uniform-convergence}
\lim_{k \to \infty} \left\Vert \M^k_{A\to B} - \N_{A\to B}\right\Vert_{\diamond} = 0~, 
\end{align}
where $\left\Vert \L_{A\to B} \right\Vert_{\diamond}$ denotes the diamond norm of a Hermiticity preserving linear map $\L_{A\to B}$,  defined as 
\begin{equation}
\left\Vert \L_{A\to B} \right\Vert_{\diamond} \equiv \sup_{\psi_{RA}\in \D(\H_R \otimes \H_A)} \left\Vert (\I_R \otimes \L_{A\to B}) (\psi_{RA})\right\Vert_1, 
\end{equation} 
where $\I_R$ is an identity channel acting on Hilbert space $\H_R$, $\psi_{RA}$ is a pure state, and system $R$ is isomorphic to the channel input system $A$ \cite{K97}. Due to the supremum being taken, note that the diamond norm might only be achieved in the limit (for example, for a sequence of two-mode squeezed vacuum states with squeezing strength becoming arbitrarily large, as discussed in \cite{M18}).

The channel sequence $\{\M^k_{A\to B}\}_k$ converges to another quantum channel $\N_{A\to B}$ in the strong sense if for all $\psi_{RA} \in \D(\H_R \otimes \H_A)$, the following holds:
\begin{align}\label{def:strong-convergence}
\lim_{k\to \infty} \left\Vert \M^k_{A\to B}(\psi_{RA}) -\N_{A\to B} (\psi_{RA}) \right\Vert_1= 0~,
\end{align}
which can be summarized more compactly as
\begin{align}
\label{def:strong-convergence-compact}
\sup_{\psi_{RA}\in \D(\H_{RA})} \lim_{k\to \infty} \left\Vert \M^k_{A\to B}(\psi_{RA}) -\N_{A\to B} (\psi_{RA}) \right\Vert_1= 0~,
\end{align}
where it is implicit that the identity channel acts on the reference system $R$.
Therefore, convergence in the strong sense is the statement that, for each fixed input quantum state $\psi_{RA}$, the sequence $\{\M^k_{A\to B}(\psi_{RA})\}_k$ of states converges to the state $\N_{A\to B}(\psi_{RA})$ in trace norm. It is important to note that the different orders in which the limits and suprema are taken in \eqref{def:uniform-convergence} and \eqref{def:strong-convergence-compact} lead to physically distinct situations, as discussed in \cite{M18}.

Let $H_A$ denote an energy observable corresponding to the quantum system $A$. Then the channel sequence $\{\M^k_{A\to B}\}_k$ converges uniformly (on the set of density operators whose marginals on the channel input have bounded energy) to another quantum channel $\N_{A\to B}$ if the following holds for some $E\in[0,\infty)$:
\begin{align}
\lim_{k \to \infty} \left\Vert \M^k_{A\to B} - \N_{A\to B}\right\Vert_{\diamond E} = 0~,
\end{align}
where the Shirokov--Winter energy-constrained diamond distance is defined as \cite{S18, W17}
\begin{multline}\label{eq:energy-constrained-diamond-norm}
\left\Vert \M^k_{A\to B} - \N_{A\to B}\right\Vert_{\diamond E} \equiv \\
 \sup_{\psi_{RA}: \tr(H_A \psi_{A}) \leq E} \left\Vert \M^k_{A\to B}(\psi_{RA}) - \N_{A\to B}(\psi_{RA}) \right\Vert_1~,  
\end{multline}
and it is again implicit that the identity channel acts on the reference system $R$. 

The energy-constrained sine distance between two quantum channels $\N_{A\to B}$ and $\M_{A\to B}$ is defined for $E \in [0,\infty)$ as \cite[Section~12]{SWAT}
\begin{multline}
C_E(\N_{A\to B}, \M_{A\to B}) \equiv  \\
  \sup_{\psi_{RA}: \tr(H_A \psi_{A}) \leq E} \sqrt{1- F(\N_{A\to B} (\psi_{RA}), \M_{A\to B} (\psi_{RA}))}.
\end{multline}

\section{Approximation of a displacement operator}\label{sec:displacement-approx}

We now analyze convergence of the experimental implementation of a displacement operator from \cite{P96} (see also \cite{DS95}) to the ideal displacement operator. For a single-mode light field, a unitary displacement operator  is defined as \cite{AS17}
 \begin{align} 
 D(\alpha) \equiv \exp(\alpha \hat{a}^{\dagger} - \alpha^{*}\hat{a}), 
 \end{align}
 where $\alpha\in \mathbb{C}$, $\hat{a} = (\hat{x} + i \hat{p})/\sqrt{2}$ is an annihilation operator, and $\hat{x}$ and $ \hat{p}$ are position- and momentum-quadrature operators, respectively. The action of a displacement operator on a single-mode Gaussian state $\rho$ can be understood as a displacement of the mean values $\langle \hat{x}\rangle_{\rho}$ and $\langle \hat{p}\rangle_{\rho}$. 
 Moreover, any displacement operator acting on $n$ modes can be decomposed as a tensor product of displacement operators acting on each mode \cite{AS17}.

 Let $\rho_A$ be a single-mode input quantum state. We then simulate the action of $D(\alpha)$ on the state $\rho_A$, according to \cite{P96}, by employing a beamsplitter $\B^{\eta}_{AB}$ of transmissivity $\eta \in (0, 1)$ and an environment state prepared in a coherent state $\ket{\beta}_{B}$ \cite{AS17}, where $\beta$ is chosen such that $\sqrt{1-\eta} \beta = \alpha$. We denote the channel corresponding to the experimental implementation of the displacement operator $D(\alpha)$ by
\begin{equation}
 \widetilde{\D}^{\eta, \beta} = \widetilde{\D}^{\eta, \frac{\alpha}{\sqrt{1-\eta}}}.
\end{equation}
As described in Figure \ref{fig:disp-sim}, the simulation of the ideal channel $\D^{\alpha}(\rho_A) \equiv D(\alpha) \rho_A D(-\alpha)$ realized by the displacement operator $D(\alpha)$ is given by the following transformation:
\begin{align}
\widetilde{\D}^{\eta, \beta}(\rho_A)\equiv \tr_B (\B^{\eta}_{AB} (\rho_A \otimes \ket{\beta}\bra{\beta}_B) ). 
\end{align}

We first show that the fidelity between the ideal displacement and its experimental approximation when acting on a fixed input state is equal to the fidelity between a pure-loss channel and an ideal channel when acting on the same input state. By using the following covariance of the beamsplitter channel with respect to displacements \cite{AS17}:
\begin{equation}
\B^{\eta}_{AB} \circ \D_B^{\beta} = \left[\D_{A}^{
	\sqrt{1-\eta}\beta}\otimes\D_{B}^{\sqrt{\eta}\beta}\right]\circ
\mathcal{B}_{AB}^{\eta}, 
\end{equation}
we arrive at the following simplification:
\begin{align} \label{eq:displacement-simulation}
(\tr_B \circ \B^{\eta}_{AB}) (\rho_A \otimes \ket{\beta}\bra{\beta}_B) = (\D_A^{\alpha}\circ \L^{\eta}_A)(\rho_A)~,
\end{align}
where $\L_{A}^{\eta}(\rho_{A})=(\tr_{B}\circ\B_{AB}^{\eta})(\rho_{A}\otimes\ket{0}\bra{0}_{B})$ denotes a pure-loss channel with transmissivity $\eta$ and $\ket{0}$ denotes the vacuum state.

\begin{figure}[ptb]
	\begin{center}
		\includegraphics[
		width=2.5in
		]{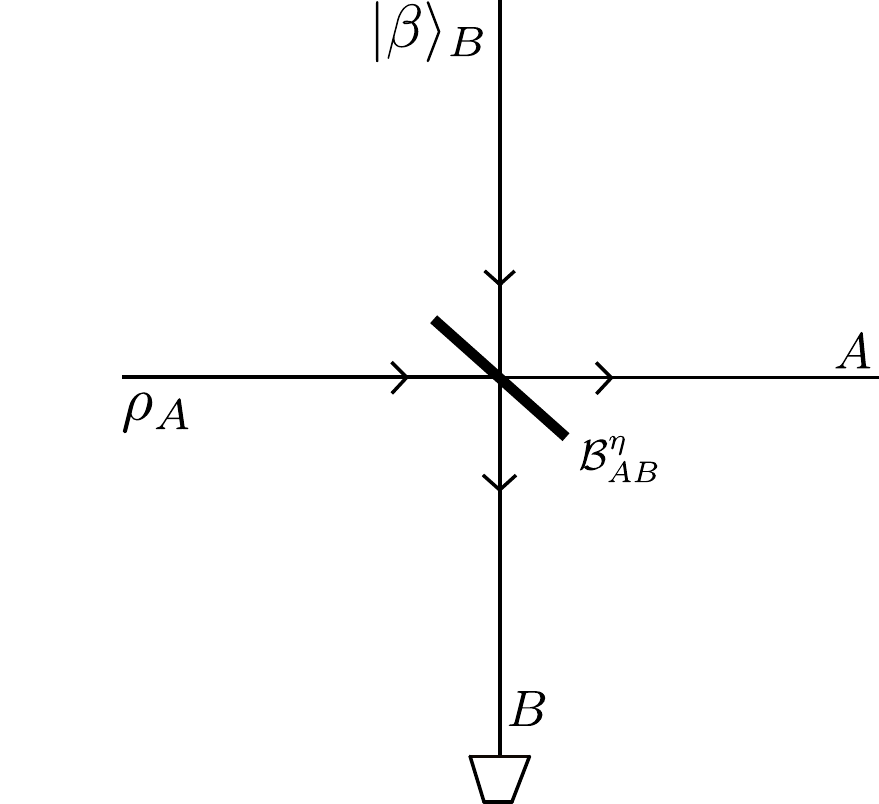}
	\end{center}
	\caption{The figure plots an experimental approximation $\widetilde{\D}^{\eta, \frac{\alpha}{\sqrt{1-\eta}}}$ of the ideal displacement operation $\D^{\alpha}$ on the input state $\rho_A$, as introduced in \cite{P96}. $\ket{\beta}_B$ represents a coherent state in mode $B$, where $\alpha = \sqrt{1-\eta}\beta$. $\B^{\eta}_{AB}$ represents a beamsplitter channel with transmissivity $\eta$. The experimental approximation of $\D^{\alpha}$ corresponds to sending $\rho_A$ and $\ket{\beta}_B$ through $\B^{\eta}_{AB}$, and then tracing out the mode $B$ \cite{P96}.}
	\label{fig:disp-sim}
	\end{figure}

Now let $\psi_{RA}$ denote an arbitrary two-mode pure state. Computing the fidelity between the ideal displacement~$\D^{\alpha}$ and its experimental approximation $\widetilde{\D}^{\eta, \frac{\alpha}{\sqrt{1-\eta}}}$ by using  \eqref{eq:displacement-simulation} and the unitary invariance of the fidelity, we find that 
\begin{equation}\label{eq:displacement-fidelity-simplified}
F(\D^{\alpha}(\psi_{RA}),\widetilde{\D}^{\eta, \frac{\alpha}{\sqrt{1-\eta}}}(\psi_{RA})) =  F(\psi_{RA}, \L^{\eta}_A(\psi_{RA})).
\end{equation}
Therefore, analyzing the convergence of the sequence $\{\widetilde{\D}^{\eta, \frac{\alpha}{\sqrt{1-\eta}}}\}_{\eta \in[0,1)}$ to $\D^{\alpha}$ is equivalent to analyzing the convergence of a sequence of pure-loss channels to an ideal channel. 

\subsection{Lack of uniform convergence}\label{sec:uniform-conv-displacement}

We now prove that the sequence  $\{\widetilde{\D}^{\eta, \frac{\alpha}{\sqrt{1-\eta}}}\}_{\eta \in[0,1)}$ does not  converge uniformly  to $\D^{\alpha}$, which follows from \eqref{eq:displacement-fidelity-simplified} and \cite[Proposition~2]{W17}.  Let $\ket{\delta}$ be a pure input coherent state. Then we find that
\begin{multline}
F(\D^{\alpha}(\ket{\delta}\bra{\delta}), \widetilde{\D}^{\eta, \frac{\alpha}{\sqrt{1-\eta}}}(\ket{\delta}\bra{\delta})) \\
= \exp[- (\vert \delta \vert^2(1-\sqrt{\eta})^2)/2],
\end{multline}
where we used \eqref{eq:displacement-fidelity-simplified} and the fact that $\vert \langle \gamma | \delta \rangle|^2 = \exp(-\left|\gamma - \delta\right|^2)$ for coherent states $\vert \gamma\rangle$ and $\vert \delta\rangle$. 
Therefore,
\begin{align}\label{eq:uniform-convergence-displacement}
\lim_{\vert\delta \vert^2 \to \infty} F(\D^{\alpha}(\ket{\delta}\bra{\delta}), \widetilde{\D}^{\eta, \frac{\alpha}{\sqrt{1-\eta}}}(\ket{\delta}\bra{\delta})) = 0~.
\end{align}

Let $\ket{\phi}_{RA} = \ket{0}_R\ket{\delta}_A$. 
Using \eqref{eq:powers}, \eqref{eq:uniform-convergence-displacement}, and the fact that $\left\Vert \rho \otimes \omega - \sigma \otimes \omega\right\Vert_1 = \left\Vert \rho - \sigma \right\Vert_1$, for any density operators $\rho, \sigma, \omega$, we find that 
\begin{align}
\label{eq:uniform-convergence-displacement-trace-distance}
\lim_{\vert\delta \vert^2 \to \infty} \left\Vert  \I_R \otimes \D_A^{\alpha}(\phi_{RA}) - \I_R \otimes \widetilde{\D}_A^{\eta, \frac{\alpha}{\sqrt{1-\eta}}}(\phi_{RA}) \right\Vert_1 = 2~,
\end{align}
which is the maximum value of the diamond distance between any two quantum channels. Therefore, the definition in \eqref{def:uniform-convergence} and the equality in \eqref{eq:uniform-convergence-displacement-trace-distance} imply that the sequence $\{\widetilde{\D}^{\eta, \frac{\alpha}{\sqrt{1-\eta}}}\}_{\eta \in[0,1)}$ does not converge uniformly to the ideal displacement channel $\D^{\alpha}$. The equality in \eqref{eq:uniform-convergence-displacement-trace-distance} indicates that the ideal displacement $\D^{\alpha}$ and its experimental approximation $\widetilde{\D}^{\eta, \frac{\alpha}{\sqrt{1-\eta}}}$ become perfectly distinguishable in the limit that the input state has unbounded energy. We note that the lack of uniform convergence of a sequence of pure-loss channels to another pure-loss channel was recently studied in \cite[Proposition~2]{W17}. 

\subsection{Strong convergence}\label{sec:strong-conv-displace}

We now argue that the sequence $\{\widetilde{\D}^{\eta, \frac{\alpha}{\sqrt{1-\eta}}}\}_{\eta \in[0,1)}$  converges to $\D^{\alpha}$ in the strong sense. Let $\chi_{\rho_A}(x, p)$ denote the Wigner characteristic function \cite{AS17} for the input state $\rho_A$. 
Let $\tilde{\rho}^{\text{out}}_A$ denote the state after the action of $\widetilde{\D}^{\eta, \frac{\alpha}{\sqrt{1-\eta}}}$ on $\rho_{A}$:
\begin{equation}
\tilde{\rho}^{\text{out}}_A = \widetilde{\D}^{\eta, \frac{\alpha}{\sqrt{1-\eta}}}(\rho_{A})\, .
\end{equation}
Then the characteristic function of $\tilde{\rho}^{\text{out}}_A$ is given by
\begin{multline}
\chi_{\tilde{\rho}^{\text{out}}_A} (x, p)=  \\\chi_{\rho_A}(\sqrt{\eta}x, \sqrt{\eta}p) e^{[i \sqrt{2} (  p\text{Re}(\alpha)-x\text{Im}(\alpha) ) -(1/4)(x^2+p^2)(1-\eta)]}\, .
\end{multline}
Moreover, the characteristic function after the action of an ideal displacement channel $\D^\alpha$ on $\rho_A$ is given by 
\begin{equation}
\chi_{\D^{\alpha}(\rho_A)} (x, p)= \chi_{\rho_A}(x, p) e^{[i \sqrt{2} ( p\text{Re}(\alpha)-x\text{Im}(\alpha))]}\, .
\end{equation}
Therefore, for each $\rho_A\in \D(\H_A)$, and for all $x, p \in \mathbb{R}$
\begin{align} \label{eq:strong-conv-displacement} 
\lim_{\eta \to 1} \chi_{\tilde{\rho}^{\text{out}}_A} (x, p) = \chi_{\D^{\alpha}(\rho_A)}(x, p)~.
\end{align} 
We have thus shown that the sequence of characteristic functions $\chi_{\tilde{\rho}^{\text{out}}_A}$ converges pointwise to $\chi_{\D^\alpha(\rho_A)}$, which implies
by \cite[Lemma~8]{LSW18}
that the sequence $\{\widetilde{\D}^{\eta, \frac{\alpha}{\sqrt{1-\eta}}}\}_{\eta \in[0,1)}$  converges to $\D^{\alpha}$ in the strong sense. 

\subsection{Convergence in the Shirokov--Winter energy-constrained diamond norm}\label{sec:SW-energy-constrained-dis-op}

We now discuss uniform convergence of the sequence $\{\widetilde{\D}^{\eta, \frac{\alpha}{\sqrt{1-\eta}}}\}_{\eta \in[0,1)}$  to $\D^{\alpha}$ on the set of density operators whose marginals on the channel input have bounded energy. As observed in \cite{S18}, a sequence of quantum channels converges strongly to a quantum channel if and only if  it converges uniformly on the set of density operators whose marginals on the channel input have bounded energy. Therefore, the sequence $\{\widetilde{\D}^{\eta, \frac{\alpha}{\sqrt{1-\eta}}}\}_{\eta \in[0,1)}$  converges uniformly to $\D^{\alpha}$ if the input states have a finite energy constraint.

However, from an experimental perspective, it is important to know how the energy-constrained uniform convergence depends on experimental parameters. Using \eqref{eq:powers} and \eqref{eq:displacement-fidelity-simplified}, we find that 
\begin{align}
\frac{1}{2}&\left\Vert\D^{\alpha} -\widetilde{\D}^{\eta, \frac{\alpha}{\sqrt{1-\eta}}}\right\Vert_{\diamond E} \nonumber \\
& \leq \sup_{\psi_{RA}: \tr(H_A\psi_{A})\leq E}\sqrt{1- F[\psi_{RA}, \L^{\eta}_A(\psi_{RA})]}  \\
 &=\sqrt{1 - \left[(1- \{E\}) \sqrt{\eta}^{\lfloor E\rfloor} + \{E\}\sqrt{\eta}^{\lceil E\rceil} \right]^2} ~, \label{eq:buresd-displacement}
\end{align}
where $\{E \} = E - \lfloor E \rfloor$. The equality follows from the recent result of \cite{N18} (see also the earlier result in \cite{RN11}), where  the energy-constrained Bures distance \cite{Sh16} between two pure-loss channels was calculated.
From \eqref{eq:buresd-displacement}, it is easy to see that 
\begin{align}
\lim_{\eta \to 1} \frac{1}{2} \left\Vert\D^{\alpha} -\widetilde{\D}^{\eta, \frac{\alpha}{\sqrt{1-\eta}}}\right\Vert_{\diamond E} = 0~,
\end{align} 
which justifies the energy-constrained uniform convergence of $\{\widetilde{\D}^{\eta, \frac{\alpha}{\sqrt{1-\eta}}}\}_{\eta \in[0,1)}$  to $\D^{\alpha}$. Furthermore, the optimal state $\psi_{RA}$  that saturates the equality in \eqref{eq:buresd-displacement} is
\begin{align}
\ket{\psi}_{RA} = \sqrt{1-\{E\}} \ket{\lfloor E\rfloor}_A \ket{\tau}_R + \sqrt{\{E\}} \ket{\lceil E \rceil}_A\ket{\tau^{\perp}}_R~, \label{eq:opt-state}
\end{align}
which follows directly from \cite{N18}. Here $\ket{\tau}$ and $\ket{\tau^{\perp}}$ are normalized orthogonal states.

Next, we perform numerical evaluations to see how close the experimental approximation $\widetilde{\D}^{\eta, \frac{\alpha}{\sqrt{1-\eta}}}$ is to the ideal displacement channel $\D^{\alpha}$. We denote the energy-constrained sine distance \cite[Section~12]{SWAT} obtained in \eqref{eq:buresd-displacement} as
\begin{align}\label{eq:energy-constrained-sd-displacement}
f(\eta, E) = \sqrt{1 - \left[(1- \{E\}) \sqrt{\eta}^{\lfloor E\rfloor} + \{E\}\sqrt{\eta}^{\lceil E\rceil} \right]^2}~.
 \end{align}
In Figure \ref{fig:energy-constraind-displacement}, we plot $f(\eta, E)$ versus $\eta$ for certain values of the energy constraint $E$. In particular, we find that for all values of $E$, the experimental approximation $\widetilde{\D}^{\eta, \frac{\alpha}{\sqrt{1-\eta}}}$ simulates 
the ideal displacement $\D^{\alpha}$ with a  high accuracy for $\eta \approx 1$. Moreover, for a fixed value of~$\eta$, the simulation of $\D^{\alpha}$ is more accurate for low values of the energy constraint on input states. 

In Figure \ref{fig:energy-constraind-displacement2}, we zoom in on Figure~\ref{fig:energy-constraind-displacement} for high values of $\eta$. Figure~\ref{fig:energy-constraind-displacement2} indicates that it is only for low values of $E$ and high values of $\eta$ that  high accuracy in simulating $\D^{\alpha}$ can be achieved. 
 Therefore, energy constraints on the input states play a critical role in simulating ideal unitary operations and determining error propagation.

\begin{figure}[ptb]
		\includegraphics[
		width=3.4in
		]{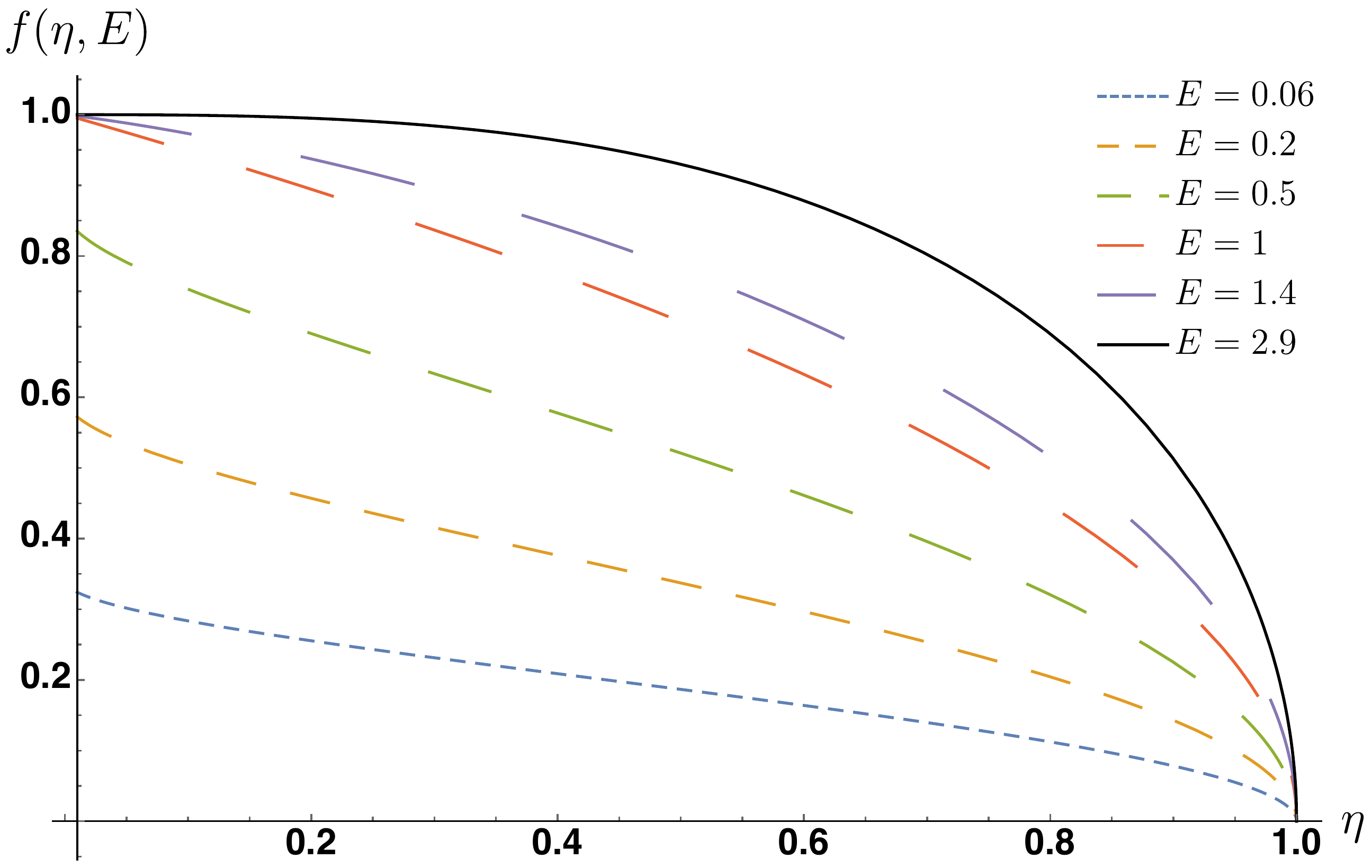}
	\caption{The figure plots the energy-constrained sine distance $f(\eta, E)$ \eqref{eq:energy-constrained-sd-displacement} between an ideal displacement channel $\D^{\alpha}$ and its experimental approximation $\widetilde{\D}^{\eta, \frac{\alpha}{\sqrt{1-\eta}}}$. In the figure, we select certain values of the energy constraint $E$, with the choices indicated next to the figure. In all the cases, $\widetilde{\D}^{\eta, \frac{\alpha}{\sqrt{1-\eta}}}$ simulates $\D^{\alpha}$ with a high accuracy for values of $\eta \approx 1$. Moreover, for a fixed value of $\eta$, the simulation of $\D^{\alpha}$ is more accurate for low values of the energy constraint on input states. }
	\label{fig:energy-constraind-displacement}\end{figure}

\begin{figure}[ptb]
	\includegraphics[
	width=3.4in
	]{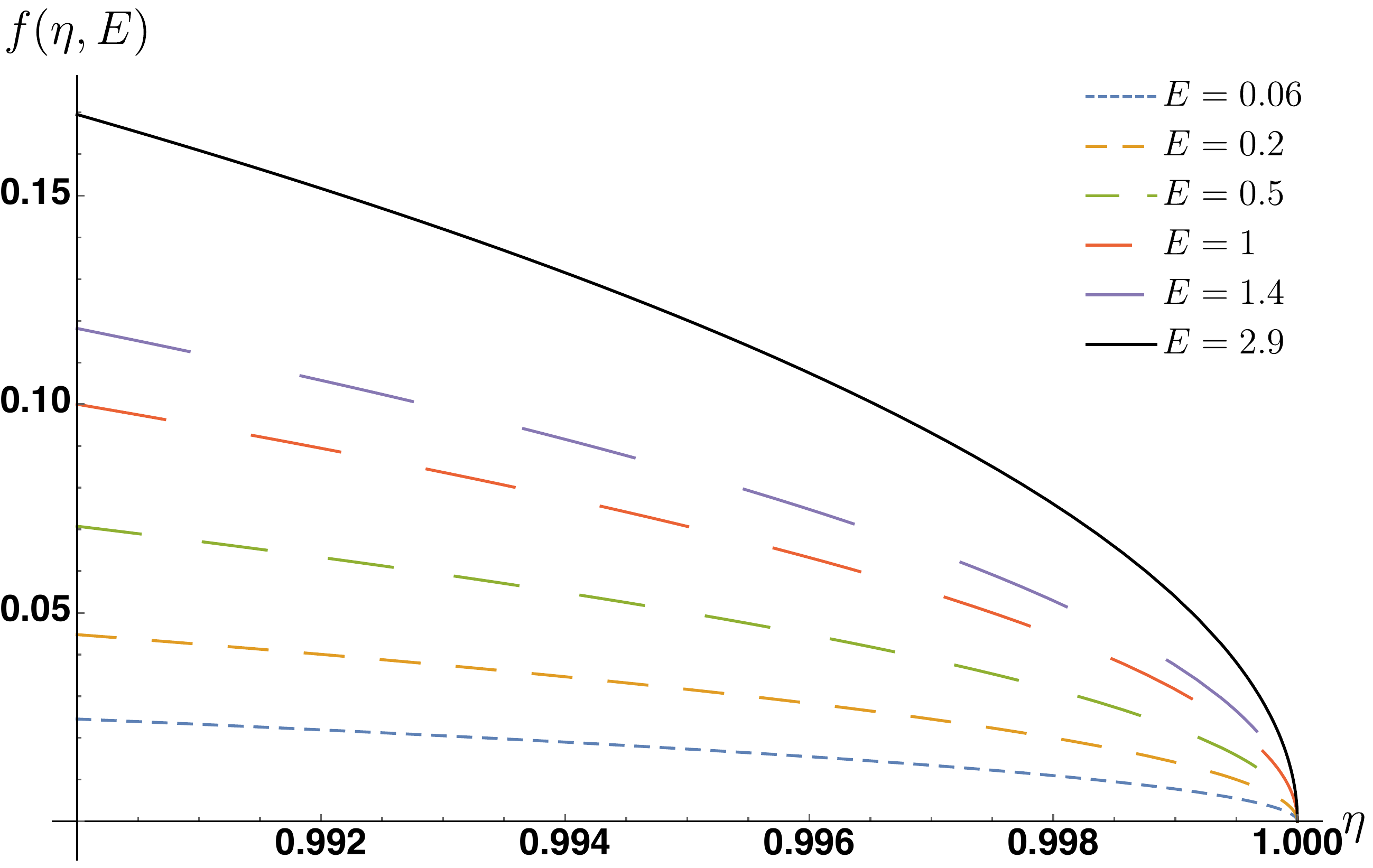}
	\caption{The figure plots  Figure \ref{fig:energy-constraind-displacement} for high values of $\eta$. The figure indicates that, only for low values of $E$ and high values of $\eta$,  high accuracy in simulating $\D^{\alpha}$ can be achieved. \label{fig:energy-constraind-displacement2}
}\end{figure}

We now analyze a simple case when the energy constraint $E$ on the input density operators takes on an integer value. From \eqref{eq:buresd-displacement}, we find that
\begin{equation}
\inf_{\psi_{RA}: \tr(H_A\psi_{A})\leq E}F[\D^{\alpha}(\psi_{RA}), \widetilde{\D}^{\eta, \frac{\alpha}{\sqrt{1-\eta}}}(\psi_{RA})] =\eta^{E} ~. \label{eq:fidelity-integerE-displacement}
\end{equation}
Therefore, for a given energy constraint on input states, and to implement an ideal displacement channel $\D^{\alpha}$ with any desired accuracy, one can find $\eta$ from \eqref{eq:buresd-displacement}--\eqref{eq:fidelity-integerE-displacement}, and the corresponding $\beta$ from $\sqrt{1-\eta}\beta  = \alpha$. The equality in 
\eqref{eq:fidelity-integerE-displacement} illustrates just how difficult it is to achieve a good accuracy in simulating an ideal displacement channel: in order to achieve the same fidelity, one requires an exponential increase in $\eta$ to match only a linear increase in~$E$.

We now summarize the results from Sections \ref{sec:uniform-conv-displacement}--\ref{sec:SW-energy-constrained-dis-op}. From Sections \ref{sec:uniform-conv-displacement} and \ref{sec:strong-conv-displace}, it follows that the sequence $\{\widetilde{\D}^{\eta, \frac{\alpha}{\sqrt{1-\eta}}}\}_{\eta \in[0,1)}$  does not converge uniformly to $\D^{\alpha}$. Rather, convergence occurs in the strong sense.  In other words, convergence of $\{\widetilde{\D}^{\eta, \frac{\alpha}{\sqrt{1-\eta}}}\}_{\eta \in[0,1)}$ to $\D^{\alpha}$ is not independent of the input state; i.e., there exists an input state for which the experimental implementation of a displacement operation has the maximum possible value of the worst-case error.

It is important to stress that, although for a fixed finite value of the energy-constraint parameter $E$, the limit $\eta \to 1$ is necessary for the implementation of a displacement operation $\D^{\alpha}$  using $\{\widetilde{\D}^{\eta, \frac{\alpha}{\sqrt{1-\eta}}}\}_{\eta \in[0,1)}$ with a high accuracy, it also relies on the fact that $\sqrt{1-\eta}\beta = \alpha$. Due to the unitary invariance of the fidelity as shown in \eqref{eq:displacement-fidelity-simplified} of our paper, the fidelity between $\D^{\alpha}$ and $\{\widetilde{\D}^{\eta, \frac{\alpha}{\sqrt{1-\eta}}}\}_{\eta \in[0,1)}$ becomes independent of the parameter $\beta$. However, it is implicit from $\sqrt{1-\eta}\beta = \alpha$  that $\eta \to 1$ requires $\beta \to \infty$. Although high values of  $\beta$ are experimentally achievable, the ideal displacement operation is achieved only in the limiting sense. This raises a further question: is it possible to implement an ideal displacement operation through a different procedure than  in \cite{P96}, such that a high accuracy can be achieved?

\subsection{Convergence for a tensor product of displacements}\label{sec:conv-tensor-prod-displacement}

Let us briefly discuss the various notions of convergence for experimental approximations of a tensor product of ideal displacement channels. Let $\{\D^{\alpha_i}\}_{i=1}^L$ be a set of $L$ different displacement channels. We  approximate the tensor product of these operators by a tensor product of $\{\widetilde{\D}^{\eta_i, \beta_i}\}_{i=1}^L$,  such that $\sqrt{1-\eta_i}\beta_i = \alpha_i$, for $i\in\{1, \dots, L\}$. From the same counterexample given above (coherent states with large energy), it follows directly that the sequence $\{\bigotimes_{i=1}^{L}\widetilde{\D}^{\eta_i, \alpha_i/\sqrt{1-\eta_i}}\}_{\eta_1, \dots, \eta_L \in [0,1)}$ does not converge uniformly to $\bigotimes_{i=1}^L\D^{\alpha_i}$. Rather, the convergence holds in the strong sense, as a consequence of \cite[Proposition~1]{M18}. Moreover, suppose that there is an average energy constraint on the input state to the tensor product of displacement operators, i.e., $\tr(\widetilde{H}_{A^L}\psi_{A^L} ) \leq E$, where
\begin{equation}
\widetilde{H}_{A^L} \equiv H_A\otimes I \otimes \dots \otimes I+\dots + I\otimes \dots \otimes I \otimes H_A,
\end{equation}
 and $E \in [0,\infty)$. Let $\tr(H_A \psi_{A_i})= E_i$, where $E_i \in [0,\infty)$, $\forall i\in\{1, \dots, L\}$. Then by using triangle inequality for the sine distance, monotonicity of the sine distance, \cite[Proposition~1]{M18} , and \eqref{eq:buresd-displacement}, we find that
\begin{multline}
\frac{1}{2} \left\Vert \bigotimes_{i=1}^L\D^{\alpha_i} -  \bigotimes_{i=1}^{L}\widetilde{\D}^{\eta_i, \alpha_i/\sqrt{1-\eta_i}}\right\Vert_{\diamond E}\leq   \\
\max_{\{E_i\}_i: \sum_i E_i \leq E} \sum_{i=1}^{L}
f(\eta_i, E_i)
\end{multline}
See the appendices for more details. Therefore, $\{\bigotimes_{i=1}^{L}\widetilde{\D}^{\eta_i, \alpha_i/\sqrt{1-\eta_i}}\}_{\eta_1, \dots, \eta_L \in [0,1)}$ converges uniformly to $\bigotimes_{i=1}^L\D^{\alpha_i}$ on the set of density operators whose marginals on the channel input have bounded energy. 

\subsection{Estimates of Shirokov--Winter energy-constrained diamond distance}

We now provide good estimates of the Shirokov--Winter energy-constrained diamond distance, as defined in \eqref{eq:energy-constrained-diamond-norm} between the ideal displacement operation $\D^{\alpha}$ and its experimental approximation $\widetilde{\D}^{\eta, \alpha/\sqrt{1-\eta}}$. In particular, we find two different lower bounds on the Shirokov--Winter energy-constrained diamond distance  by using two different techniques. A first technique is based on the trace distance between $\I_R \otimes \D^{\alpha}(\psi_{RA})$ and $\I_R\otimes \widetilde{\D}^{\eta, \alpha/\sqrt{1-\eta}} (\psi_{RA})$ for a finite energy-constraint $E$, i.e., $\tr(\hat{n} \psi_A  )\leq E$,  where $\psi_{RA}$ is given by \eqref{eq:opt-state}. Since the Shirokov--Winter energy-constrained diamond distance, as defined in \eqref{eq:energy-constrained-diamond-norm} involves an optimization over all input states satisfying the energy constraint, we find that
\begin{multline}\label{eq:td-diamond-d}
\left\Vert \I_R \otimes \D^{\alpha}(\psi_{RA})  - \I_R\otimes \widetilde{\D}^{\eta, \frac{\alpha}{\sqrt{1-\eta}}} (\psi_{RA}) \right\Vert_1  \\
\leq \left\Vert \D^{\alpha} -  \widetilde{\D}^{\eta, \frac{\alpha}{\sqrt{1-\eta}}} \right\Vert_{\diamond E}~.
\end{multline}

A second technique is based on the numerical evaluation of the Shirokov--Winter energy-constrained diamond distance between $\D^{\alpha}$ and  $\widetilde{\D}^{\eta, \alpha/\sqrt{1-\eta}}$ on a truncated Hilbert space. In particular, we consider input states to these quantum channels such that instead of acting on an infinite-dimensional separable Hilbert space, these states act on an $M$-dimensional Fock space. Moreover, we consider a mean photon number constraint on these states. Let $\H_M$ denote an $M$-dimensional Fock space.   Let  $\underline{\hat{n}}$ denote  the following truncated number operator:
\begin{align}\label{eq:num-opt-truncated}
\underline{\hat{n}} = \sum_{n=0}^{M}  n \ket{n}\bra{n}~. 
\end{align}
Let $\varphi_A \in \D(\H_M)$. 
Then the following inequality holds: 
\begin{align}\label{eq:en-constr-truncated}
\tr(\underline{\hat{n}} \varphi_A )\leq E~,
\end{align}
where $E$ denotes the mean energy constraint.

We define the energy-constrained diamond distance between two quantum channels $\N_{A\to B}$ and $\M_{A\to B}$ on a truncated Hilbert space as
\begin{multline}\label{eq:diam-dist_truncated}
\left\Vert \N - \M \right\Vert_{\diamond E, M}\\
\equiv\sup_{\phi_{RA}\in \D(\H_M^{\otimes 2}): \tr(\underline{\hat{n}} \phi_A )\leq E  } \left\Vert \N(\phi_{RA}) - \M (\phi_{RA}) \right\Vert_1~,
\end{multline}
where $E$ and $M$ denote the mean energy constraint and the truncation parameter, respectively, and $\phi_{RA} = \ket{\phi}\bra{\phi}_{RA}$ is a purification of the state $\phi_A$. Moreover, it is implicit that the identity channel acts on the reference system $R$. Note that the following identity holds
\begin{multline}\label{eq:diam-dist_truncated-extra}
\left\Vert \N - \M \right\Vert_{\diamond E, M}\\
=\sup_{\phi_{RA}\in \D(\H_M^{\otimes 2}): \tr(\hat{n} \phi_A )\leq E  } \left\Vert \N(\phi_{RA}) - \M (\phi_{RA}) \right\Vert_1~,
\end{multline}
where we have replaced $\underline{\hat{n}}$ with $\hat{n}$, following as a consequence of the reduced state of $\phi_{RA}\in \D(\H_M^{\otimes 2})$ on $A$ having support only on the truncated space and from the Schmidt decomposition, implying that the reference system $R$ need only have support as large as the input space $A$.

We now show that the set of  density operators  acting on a truncated Hilbert space with a finite mean energy constraint (yet an arbitrarily high truncation parameter)  is dense in the set of density operators acting on an infinite-dimensional Hilbert space and with the same mean energy constraint. In other words, any finite mean-energy state acting on an infinite-dimensional separable Hilbert space can be approximated with an arbitrary accuracy by a state with the same finite mean-energy acting on a truncated Hilbert space with a sufficiently high value of the truncation parameter. Let $\rho_{RA}$ denote a density operator acting on an infinite-dimensional separable Hilbert space, such that $\tr(\hat{n}_A \rho_{RA}) \leq E$, where $E>0$.  Let $\Pi_A^M$ denote an $M$-dimensional projector defined as
\begin{align}
\Pi^M_A = \sum_{n=0}^{M} \ket{n}\bra{n}~. 
\end{align}	

Consider the following chain of inequalities: 
\begin{align}
\tr(\Pi^M_A \rho_{RA})  &=\tr(\rho_{RA}) -  \sum_{n=M+1}^{\infty} \langle n \vert \rho_A \vert n \rangle~ \\
& \geq 1- \sum_{n=M+1}^{\infty} \frac{n}{M+1} \langle n \vert \rho_A \vert n \rangle~\\
& \geq 1- \frac{1}{M+1}\bigg(\sum_{n =0}^{\infty} n \langle n \vert \rho_A \vert n \rangle    \bigg)\\
& \geq 1- \frac{E}{M+1}~. \label{eq:trunc-proj-expec}
\end{align}
 The first inequality follows from the fact that $n/\left(M+1\right) \geq 1$ for all  $n \in [M+1, \infty)$. The second inequality follows because $\sum_{n=0}^{M} n \langle n \vert \rho_A \vert n \rangle$ is a sum of positive numbers. The last inequality follows because 
$\tr(\hat{n}_A \rho_{A}) \leq E$. We note that \eqref{eq:trunc-proj-expec} can also be derived from the Fock cutoff lemma in \cite{KL11}.

Let $\rho^M_{RA}$ denote the following truncated state 
\begin{align}\label{eq:trun-state}
\rho^M_{RA} = \frac{\Pi^M_A \rho_{RA} \Pi^M_A}{\tr(\Pi^M_A \rho_{RA})}~. 
\end{align}
The following proposition establishes a bound on the trace distance between $\rho_{RA}$ and $\rho^M_{RA}$. 

\begin{proposition}\label{pro:gent-meas}
	Let $\rho_{RA}$ be a density operator acting on an infinite-dimensional separable Hilbert space such that $\tr(\hat{n}_A\rho_{RA})\leq E$, where $E>0$, and $\hat{n}_A$ is the number operator as defined in \eqref{eq:num-opt}. Let $\rho^M_{RA}$ be the $M$-dimensional truncation of the state $\rho_{RA}$, as defined in \eqref{eq:trun-state}. Then 
	\begin{align}
\frac{1}{2}\left\Vert \rho_{RA} - \rho^M_{RA} \right\Vert_1 \leq \sqrt{\frac{E}{M+1}}~.  
	\end{align}
\end{proposition}
\begin{proof}
The proof follows directly from \eqref{eq:trunc-proj-expec} and the gentle measurement lemma introduced in \cite{W99} and subsequently improved in \cite{ON07}. 
\end{proof}

\bigskip
Proposition~\ref{prop:diam-dis-trunc} below states that for low values of the mean energy constraint $E$, the Shirokov--Winter energy-constrained diamond distance between two quantum channels $\N$ and $\M$ can be estimated with an arbitrarily high accuracy by using the energy-constrained diamond distance on a truncated input Hilbert space with sufficiently high values of the truncation parameter $M$. 

\begin{proposition}\label{prop:diam-dis-trunc}
	Let $\N$ and $\M$ be quantum channels, and let $E$ be the energy constraint on the input states to these channels. Let $M$ denote the truncation parameter. Then 
	\begin{multline}
 \tfrac{1}{2}\left\Vert \N - \M \right\Vert_{\diamond E, M} \leq	 \tfrac{1}{2}\left\Vert \N - \M \right\Vert_{\diamond E}\\ \leq  \tfrac{1}{2}\left\Vert \N - \M \right\Vert_{\diamond E, M} + 2 \sqrt{\frac{E}{M+1}}~. 
	\end{multline}
\end{proposition}
\begin{proof}
The inequality $\left\Vert \N - \M \right\Vert_{\diamond E, M} \leq \left\Vert \N - \M \right\Vert_{\diamond E}$ follows from  \eqref{eq:energy-constrained-diamond-norm} and \eqref{eq:diam-dist_truncated-extra}.

We now prove the other inequality.  
Let $\rho_{RA}$ be a density operator acting on an infinite-dimensional 
separable Hilbert space such that $\tr(\hat{n}_A \rho_{RA})\leq E$. Let $\rho_{RA}^M$ be the $M$-dimensional truncation of the state $\rho_{RA}$ as defined in \eqref{eq:trun-state}. 
Consider the following chain of inequalities:
\begin{align}
&\left\Vert \N(\rho_{RA}) - \M(\rho_{RA})\right\Vert_1 \nonumber\\ 
&\leq \left\Vert \N(\rho_{RA}) - \N(\rho^M_{RA})\right\Vert_1+\left\Vert \N(\rho^M_{RA}) - \M(\rho^M_{RA})\right\Vert_1\nonumber \\
& \qquad \qquad \qquad \qquad+\left\Vert \M(\rho^M_{RA}) - \M(\rho_{RA})\right\Vert_1\\
& \leq 2 \left\Vert \rho_{RA} - \rho^M_{RA} \right\Vert_1 + \left\Vert \N(\rho^M_{RA}) - \M(\rho^M_{RA})\right\Vert_1\\
& \leq 4 \sqrt{\frac{E}{M+1}} + \left\Vert \N - \M\right\Vert_{\diamond E, M}~. 
\end{align}
In all the steps above, it is implicit that the identity channel acts on the reference system $R$. The first inequality is the consequence of triangle inequality for the trace distance. The second inequality follows from monotonicity of the trace distance. The last inequality follows from Proposition \ref{pro:gent-meas} and from \eqref{eq:diam-dist_truncated}. Since the chain of inequalities holds for all input states $\rho_{RA}$ satisfying the energy constraint, the desired result follows. 
\end{proof}

\bigskip
We now study the aforementioned two techniques in detail to characterize the performance of the simulation of an ideal displacement operator. 
It is evident from Figure \ref{fig:energy-constraind-displacement} that for a fixed value of $\eta$, the accuracy in simulating an ideal displacement operation $
\D^{\alpha}$ by using the protocol from \cite{P96} is reasonable only for low values of the energy constraint on input states. Therefore, we now study  the simulation of $\D^{\alpha}$ in detail only for low values of the energy constraint. 

Let $0<E<1$. Then 
\begin{multline}\label{eq:trace-distance-for-displacement}
\frac{1}{2} \left\Vert  \I_R \otimes \D^{\alpha}(\psi_{RA})  - \I_R\otimes \widetilde{\D}^{\eta, \frac{\alpha}{\sqrt{1-\eta}}} (\psi_{RA})\right\Vert_1 \\
= \frac{1}{2}[\{E\}(1-\eta) +(1-\sqrt{\eta})\varkappa(\eta, \{E\})] 
\equiv  d_1(\eta, E)~,
\end{multline}
where
\begin{equation}
\varkappa(\eta, \{E\}) = \sqrt{\{E\}(4+\{E\}(\eta+2\sqrt{\eta}-3))},
\end{equation}
$\{E\} = E- \lfloor E\rfloor$, and  $\psi_{RA}$ is given by \eqref{eq:opt-state} (see the appendices and \cite{Mathematica} for a detailed proof to obtain \eqref{eq:trace-distance-for-displacement}). Therefore, from \eqref{eq:td-diamond-d}, it follows that \eqref{eq:trace-distance-for-displacement} is a lower bound on the Shirokov--Winter energy-constrained diamond distance between $\D^{\alpha}$ and $\widetilde{\D}^{\eta, \alpha/\sqrt{1-\eta}}$ for $0<E<1$, i.e., 
\begin{align}
d_1(\eta, E) \leq \frac{1}{2}\left\Vert \D^{\alpha} -  \widetilde{\D}^{\eta, \frac{\alpha}{\sqrt{1-\eta}}} \right\Vert_{\diamond E}~.
\end{align}

As discussed earlier, a second method to obtain a lower bound on the Shirokov--Winter energy-constrained diamond distance between two quantum channels $\N_{A\to B}$ and $\M_{A\to B}$ is to truncate the infinite-dimensional separable Hilbert space to a finite-dimensional Hilbert space and apply energy constraints on channel input states according to the truncated number operator, as defined in \eqref{eq:num-opt-truncated}--\eqref{eq:en-constr-truncated}. In particular, we obtain the energy-constrained diamond distance between $\N_{A\to B}$ and $\M_{A\to B}$ on a truncated Hilbert space  by using a semi-definite program (SDP) from \cite{W17}, which is inspired from an SDP defined in the context of finite-dimensional quantum channels in  \cite{JW09,JW13}. We use the following SDP to estimate the Shirokov--Winter energy-constrained diamond distance between two quantum channels $\N_{A\to B}$ and $\M_{A\to B}$:
\begin{equation}\label{eq:SDP}
\left\Vert \N - \M \right\Vert_{\diamond E, M} = \left\{\begin{array}{l l}\sup & \tr(W_{RB} J_{RB})\\
\text{subject to} & 0 \leq W_{RB} \leq \rho_R \otimes \mathbb{I}_B ,\\
& \tr( \rho_R ) =1, \ \rho_R \geq 0, \\
& \tr( \underline{\hat{n}} \rho_R  ) \leq E,
\end{array}\right.
\end{equation}
where $M$ is the truncation parameter, $E$ is the mean energy-constraint parameter, and $\underline{\hat{n}}$ is given by \eqref{eq:num-opt-truncated}. Moreover, $J_{RB}$ denotes the operator corresponding to the difference of the Choi operators of quantum channels $\N$ and $\M$ on the truncated Hilbert space $\H_M$ and  is defined as follows
\begin{equation}
J_{RB} = (\I_R \otimes \N_{A\to B}) (\Gamma_{RA}) - (\I_R \otimes \M_{A\to B})  (\Gamma_{RA}),
\end{equation}
where $\Gamma_{RA} = \ket{\Gamma}\bra{\Gamma}_{RA}$ is the projection onto the unnormalized maximally entangled vector on the truncated Hilbert space $\H_M$, i.e., 
\begin{align}
\ket{\Gamma}_{RA} = \sum_{n=0}^M \ket{n}_R \ket{n}_A~.
\end{align}

For a small value of the energy-constraint parameter~$E$, the truncation parameter $M$ can be chosen such that the value of $\Vert \N - \M \Vert_{\diamond E, M}$ does not change significantly by increasing $M$ further. For example, in the context of the ideal displacement operation $\D^{\alpha}$ and its experimental approximation $\widetilde{\D}^{\eta, \alpha/\sqrt{1-\eta}}$, we find that for $E\ll 1$, the truncation parameter $M=6$ provides a good estimate of $\left\Vert \D^{\alpha} -  \widetilde{\D}^{\eta, \frac{\alpha}{\sqrt{1-\eta}}} \right\Vert_{\diamond E}$. For $E\ll1$, we define  
\begin{align}\label{eq:d2}
d_2(\eta, E) \equiv \frac{1}{2}\left\Vert \D^{\alpha} -  \widetilde{\D}^{\eta, \frac{\alpha}{\sqrt{1-\eta}}} \right\Vert_{\diamond E,M}~,
\end{align}
for $M=6$. From Proposition \ref{prop:diam-dis-trunc} it follows that 
\begin{align}
d_2(\eta, E)\leq  \frac{1}{2}\left\Vert \D^{\alpha} -  \widetilde{\D}^{\eta, \frac{\alpha}{\sqrt{1-\eta}}} \right\Vert_{\diamond E}.
\end{align}

\begin{figure}[ptb]
	\includegraphics[
	width=3.4in
	]{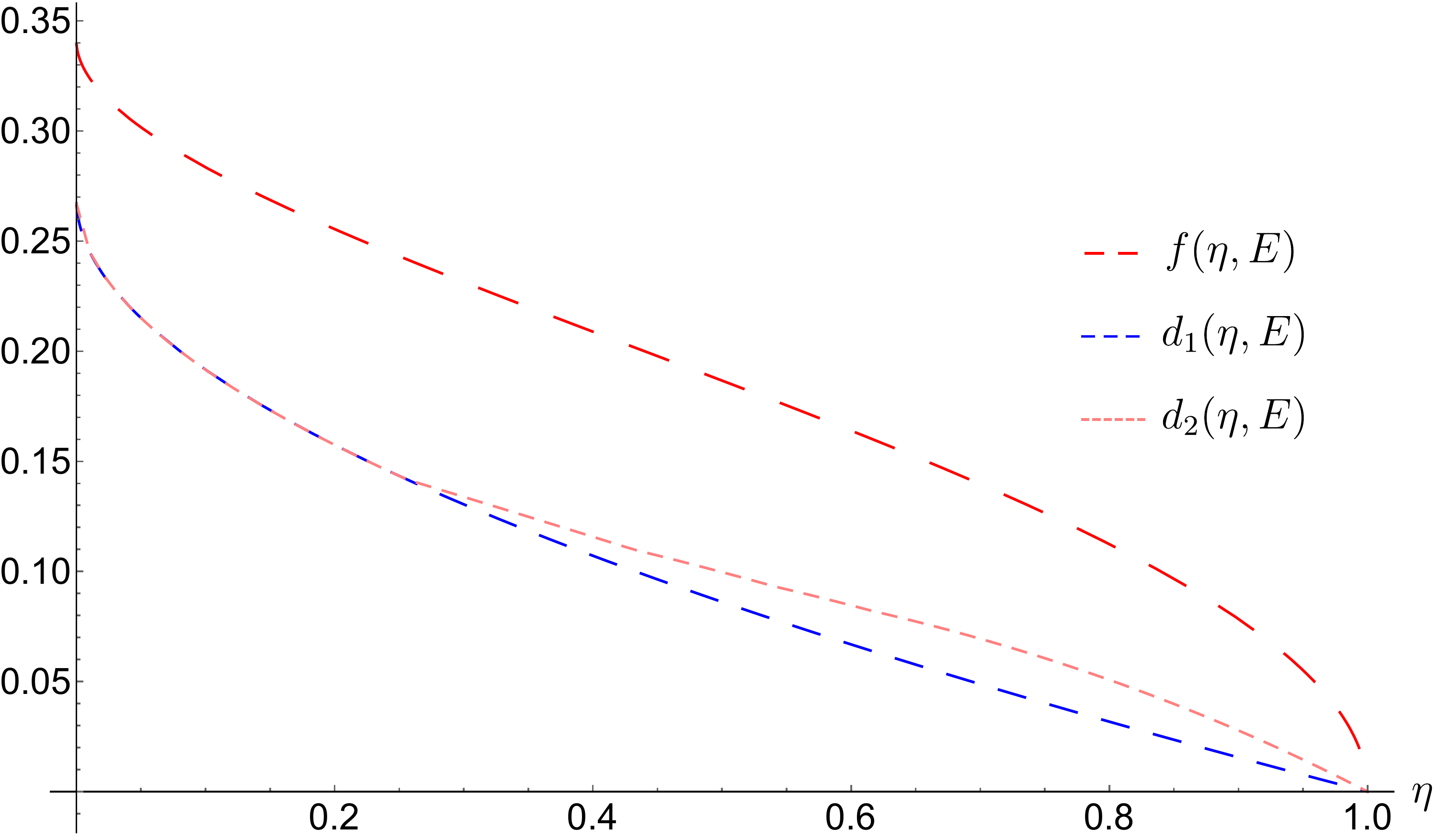}
	\caption{The figure depicts the lower bound $d_1(\eta, E)$ in  \eqref{eq:trace-distance-for-displacement}, the lower bound $d_2(\eta, E)$ in \eqref{eq:d2}, and the upper bound $f(\eta, E)$ in \eqref{eq:energy-constrained-sd-displacement} for the fixed value $E=0.06$. Here, $d_1(\eta, E)$ is the trace distance between the outputs of an ideal displacement $\D^{\alpha}(\psi_{RA})$ and its experimental approximation $\widetilde{\D}^{\eta, \alpha/\sqrt{1-\eta}}(\psi_{RA})$, when the input state $\psi_{RA}$ is such that it optimizes the energy-constrained sine distance between $\D^{\alpha}$ and $\widetilde{\D}^{\eta, \alpha/\sqrt{1-\eta}}$ and is given by \eqref{eq:opt-state}. Moreover, $d_2(\eta, E)$ is the energy-constrained diamond distance between $\D^{\alpha}$ and $\widetilde{\D}^{\eta, \alpha/\sqrt{1-\eta}}$ on a truncated Hilbert space with the truncation parameter $M=6$, and $f(\eta, E)$ is the energy-constrained sine distance between $\D^{\alpha}$ and $\widetilde{\D}^{\eta, \alpha/\sqrt{1-\eta}}$. 
	For  low values of $\eta$,	$d_1(\eta, E)$ is close to  $d_2(\eta, E)$. The figure indicates that  for a fixed value of $E$,  high accuracy in simulating $\D^{\alpha}$ can be achieved only for high values of $\eta$.
		\label{fig:comparison}
}\end{figure}

Let us study in detail the case when the input states have mean energy constraint $E =0.06$. We first calculate $d_1(\eta, E)$ by using \eqref{eq:trace-distance-for-displacement} and then find $d_2(\eta, E)$, as defined in \eqref{eq:d2} by solving the corresponding SDP in \eqref{eq:SDP} \cite{Mathematica}. We then compare both $d_1(\eta, E)$ and $d_2(\eta, E)$ with the energy-constrained sine distance $f(\eta, E)$ between $\D^{\alpha}$ and $\widetilde{\D}^{\eta, \alpha/\sqrt{1-\eta}}$, as calculated in \eqref{eq:energy-constrained-sd-displacement}.

In Figure~\ref{fig:comparison}, we plot the lower bound $d_1(\eta, E)$ in \eqref{eq:trace-distance-for-displacement}, the lower bound $d_2(\eta, E)$ in \eqref{eq:d2}, and the upper bound $f(\eta, E)$ in \eqref{eq:energy-constrained-sd-displacement} versus $\eta$ for $E=0.06$. In particular, we find that $d_1(\eta, E)$ overlaps with $d_2(\eta, E)$ for small values of $\eta$. From numerical evaluations, we find that the value of $d_2(\eta, E)$ does not change significantly with a further increment in $M\geq6$. These findings indicate that  $d_2(\eta, E)$ is a good lower bound on the Shirokov--Winter energy-constrained diamond distance between $\D^{\alpha}$ and $\widetilde{\D}^{\eta, \alpha/\sqrt{1-\eta}}$, and furthermore, that the upper bound in Proposition~\ref{prop:diam-dis-trunc} is loose for this case. Moreover, from Figure~\ref{fig:comparison}, it is evident that $d_1(\eta, E)$ is also a tight lower bound. Although there is a significant gap between $f(\eta, E)$ and $d_2(\eta, E)$ in Figure~\ref{fig:comparison}, the key message of our results remains the same; i.e., in order to achieve a high accuracy in simulating an ideal displacement operation $\D^{\alpha}$ by using the protocol from \cite{P96}, the value of $\eta$ should be very high and the mean energy of the input states should be very low. In summary, a good estimation of the accuracy in simulating an ideal displacement operation can be obtained from the following three methods:
\begin{enumerate}
	\item The energy-constrained sine distance between $\D^{\alpha}$ and $\widetilde{\D}^{\eta, \alpha/\sqrt{1-\eta}}$ can be calculated from the analytical expression obtained in \eqref{eq:energy-constrained-sd-displacement}. 
	\item A lower bound on the Shirokov--Winter energy-constrained diamond distance between $\D^{\alpha}$ and $\widetilde{\D}^{\eta, \alpha/\sqrt{1-\eta}}$ can be established by solving an SDP in \eqref{eq:SDP} on a truncated Hilbert space \cite{Mathematica} . 
	\item For a fixed energy range $\lfloor E \rfloor \leq E\leq \lceil E\rceil$, a lower bound on the Shirokov--Winter energy-constrained diamond distance between $\D^{\alpha}$ and $\widetilde{\D}^{\eta, \alpha/\sqrt{1-\eta}}$ can be established by finding the trace distance between $\I_R \otimes \D^{\alpha}(\psi_{RA})$ and $\I_R\otimes \widetilde{\D}^{\eta, \alpha/\sqrt{1-\eta}} (\psi_{RA})$, where $\psi_{RA}$ is given by \eqref{eq:opt-state}. In particular, for $0<E<1$, an analytical expression for the trace distance is given by \eqref{eq:trace-distance-for-displacement} (see the appendices for more information).
\end{enumerate}

\section{Approximation of a beamsplitter}\label{sec:bs-approx}

In this section, we analyze convergence of the experimental implementations of a beamsplitter transformation. A beamsplitter consists of a semi-reflective mirror, which both partly reflects and transmits the input radiation. In general the unitary operator corresponding to the beamsplitter transformation is given by \cite{LOQC07}
\begin{equation}\label{eq:bs-unitary}
U^{\theta, \phi}_{\text{BS}} \equiv \exp[i\theta (e^{i \phi} \hat{a}_{\text{in}}^{\dagger}\hat{b}_{\text{in}} +  e^{-i \phi} \hat{a}_{\text{in}}\hat{b}_{\text{in}}^{\dagger})],
\end{equation}
where $\hat{a}_{\text{in}}$ and $\hat{b}_{\text{in}}$ denote the two incoming modes on either side of the beamsplitter, and $\theta$ depends on the interaction time and coupling strength of semi-reflective mirrors. Moreover, $\phi$ denotes the relative phase shift parameter. 
Another representation of a beamsplitter is given in terms of the transmissivity $\eta$ of the beamsplitter, where $\eta = (\cos{\theta})^2$. In this work, we parametrize a beamsplitter with respect to $\eta$ and $\theta$ interchangeably. 

Let $\rho_{A_1A_2}$ be a two-mode input quantum state, and let $\B^{\eta, \phi}$ denote the beamsplitter transformation of transmissivity $\eta\in (0, 1)$ and phase $\phi \in [0, 2\pi]$ acting on mode $A_1$ and $A_2$, i.e., 
\begin{equation}
    \B^{\eta, \phi}(\rho_{A_1A_2}) = U^{\theta, \phi}_{\text{BS}}(\rho_{A_1A_2}) (U^{\theta, \phi}_{\text{BS}})^{\dagger},
\end{equation}
where $\eta = (\cos{\theta})^2$.

There are at least two different ways to model the noise in implementing $\B^{\eta, \phi}$. A first method consists of an ideal beamsplitter preceded and followed by a tensor product of pure-loss channels with transmissivity $\eta'$.  We denote the channel corresponding to the experimental implementation of the beamsplitter transformation $\B^{\eta, \phi}$ by 
$\widetilde{\B}^{\eta, \phi, \eta'}$, such that 
\begin{equation}\label{eq:bs-approx-1}
\widetilde{\B}^{\eta, \phi, \eta'}(\rho_{A_1A_2}) \equiv (\L^{\eta'}_{A_1} \otimes \L^{\eta'}_{A_2})\circ\B^{\eta, \phi} \circ (\L^{\eta'}_{A_1} \otimes \L^{\eta'}_{A_2})(\rho_{A_1A_2}).
\end{equation}
This models the physical process when there is a non-zero probability of absorption of the radiation, along with reflection and transmission. We note that the beam splitter channel with transmissivity $\eta$ and the channel corresponding to a tensor product of two pure-loss channels with equal transmissivity $\eta'$ commute, i.e., 
\begin{equation}
(\L^{\eta'}_{A_1} \otimes \L^{\eta'}_{A_2})\circ\B^{\eta, \phi} (\cdot) = \B^{\eta, \phi} \circ (\L^{\eta'}_{A_1} \otimes \L^{\eta'}_{A_2}) (\cdot).
\end{equation}
Since a concatenation of two pure-loss channels $\L^{\eta_1}$  and $\L^{\eta_2}$ is another pure-loss channel $\L^{\eta_1\eta_2}$ with transmissivity $\eta_1\eta_2$, we consider the following channel as an experimental approximation of the beamsplitter:
\begin{equation}
\widetilde{\B}^{\eta, \phi, \eta'}(\rho_{A_1A_2}) \equiv \B^{\eta, \phi} \circ (\L^{\eta'}_{A_1} \otimes \L^{\eta'}_{A_2})(\rho_{A_1A_2}),
\end{equation}

Now let $\psi_{RA_1A_2}$ denote an arbitrary four-mode pure state, where it is understood that $R$ is a two-mode system.  
From invariance of the fidelity under a unitary transformation, we find that
\begin{multline}\label{eq:fid-ideal-BS-noisy-BS}
F(\B^{\eta,\phi}(\psi_{RA_1A_2}), \widetilde{\B}^{\eta,\phi, \eta'}(\psi_{RA_1A_2})) \\= F(\psi_{RA_1A_2},(\L^{\eta'}_{A_1} \otimes \L^{\eta'}_{A_2})(\psi_{RA_1A_2}) ),
\end{multline}
where it is implicit that the identity channel acts on the reference system $R$. 

From arguments similar to those given in Section \ref{sec:uniform-conv-displacement}, we find that the sequence $\{\widetilde{\B}^{\eta, \phi, \eta'}  \}_{\eta' \in [0, 1)}$ does not converge uniformly to $\B^{\eta, \phi}$. In particular, let $\psi_{A_1A_2} =\ket{\delta}\bra{\delta}_{A_1} \otimes \ket{0}\bra{0}_{A_2}$ be a tensor product of a  coherent state and a vacuum state. Then we find that 
\begin{multline}
F(\B^{\eta, \phi}(\psi_{A_1A_2}), \widetilde{\B}^{\eta, \phi, \eta'}(\psi_{A_1A_2})) \\=  \exp[- (\vert \delta \vert^2 (1-\sqrt{\eta'})^2)/2] 
\end{multline}
Therefore, from arguments similar to \eqref{eq:uniform-convergence-displacement} and \eqref{eq:uniform-convergence-displacement-trace-distance}, it follows that the ideal beamsplitter $\B^{\eta, \phi}$ and its experimental approximation $\widetilde{\B}^{\eta, \phi, \eta'}$ become perfectly distinguishable in the limit that the input state has unbounded
energy. Hence the uniform convergence does not hold. 

We now argue that the sequence converges $\{\B^{\eta, \phi, \eta'}  \}_{\eta'\in [0, 1)}$ to  $\B^{\eta, \phi}$ in the strong sense. Let $\chi_{\rho_{A_1A_2}}(x_1, p_1, x_2, p_2)$ denote the Wigner characteristic function for the input state $\rho_{A_1A_2}$. 
Let $\tilde{\rho}^{\text{out}}_{A_1A_2}(\eta, \phi, \eta')$ denote the state after the action of $\widetilde{\B}^{\eta, \phi, \eta'}$ on $\rho_{A_1A_2}$. 
Then for each $\rho_{A_1A_2}\in \D(\H_{A_1}\otimes \H_{A_2})$, and for all $x_1, p_1, x_2, p_2 \in \mathbb{R}$, we have that
\begin{multline} \label{eq:strong-conv-bs1} 
\lim_{\eta' \to 1} \chi_{\tilde{\rho}^{\text{out}}_{A_1A_2}(\eta, \phi, \eta')} (x_1, p_1, x_2, p_2) \\
= \chi_{\B^{\eta, \phi}(\rho_{A_1A_2})}(x_1, p_1, x_2, p_2).
\end{multline} 
We have thus shown that the sequence of characteristic functions $\chi_{\tilde{\rho}^{\text{out}}_{A_1A_2}(\eta, \phi,  \eta')}$ converges pointwise to $\chi_{\B^{\eta,\phi}(\rho_{A_1A_2})}$, which implies
by \cite[Lemma~8]{LSW18}
that the sequence $\{\widetilde{\B}^{\eta,\phi, \eta'}\}_{\eta' \in[0,1)}$  converges to $\B^{\eta, \phi}$ in the strong sense (see Appendix \ref{app:bs} for more details).

Similar to Section \ref{sec:SW-energy-constrained-dis-op}, we investigate the dependence of convergence of the sequence of $\{\widetilde{\B}^{\eta, \phi, \eta'}  \}_{\eta' \in [0, 1)}$ to  $\B^{\eta, \phi}$ on the experimental parameters when there is a finite energy constraint on the input states. Let $A = A_1A_2$. 
Consider the following chain of inequalities:
\begin{align}
\frac{1}{2}&\Vert\B^{\eta, \phi} - \widetilde{\B}^{\eta, \phi, \eta'
} \Vert_{\diamond E} \nonumber \\
&\leq \sup_{\psi_{RA}: \tr(H_{A}\psi_{A})\leq E} C(\psi_{RA}, (\L^{\eta'}_{A_1}\otimes \L^{\eta'}_{A_2})(\psi_{RA}))\\
& = f(\eta', E),
\end{align}
where $H_{A} = H_{A_1} \otimes I_{A_2} + I_{A_1}\otimes H_{A_2} $ and  $f(\eta', E)$ is given by \eqref{eq:energy-constrained-sd-displacement}.
The first inequality follows  from \eqref{eq:powers}. The last inequality  follows from the recent result of \cite{N18}, which holds for a tensor product of loss channels with the same transmissivity. Therefore, from the analysis in Section \ref{sec:SW-energy-constrained-dis-op}, it follows that the accuracy in implementing $\B^{\eta, \phi}$ using $\B^{\eta, \phi, \eta'}$ is high only for high values of the loss parameter $\eta'$ and low values of the energy constraint $E$.

A second experimental approximation of an ideal beamsplitter is a phenomenological model that accounts for the imprecision in implementing $\B^{\theta, \phi}$ with an exact value of the parameters $\theta$ and $\phi$, as defined in \eqref{eq:bs-unitary}.  For the analysis that follows, we fix $\phi = 0$, which is typically considered in experiments. We denote the ideal beamsplitter for $\phi=0$ by $\B^{\theta}$. We note that a similar analysis follows for $\phi = \pi/2$.

The channel corresponding to an experimental approximation of $\B^{\theta}$ is given by 
\begin{equation}\label{eq:bs-approx2}
\widetilde{\B}^{\theta, \sigma} \equiv \int_{0}^{2\pi} d\theta'~p(\theta', \theta, \sigma, 0, 2\pi) \B^{\theta'},
\end{equation}
where $p(\theta',\theta,  \sigma, a , b)$ is a truncated normal distribution with location parameter $\theta$, scale parameter $\sigma$, truncation range in between $a$ and $b$, and is given by
\begin{equation}\label{eq:truncated-uni-prob}
p(\theta', \theta, \sigma, a, b) = \frac{\phi((\theta'-\theta)/\sigma)}{\sigma[\Phi((b-\theta)/\sigma)-\Phi((a-\theta)/\sigma)]},
\end{equation}
where
\begin{align}
\phi(x) = \frac{1}{\sqrt{2\pi}}\exp(-x^2/2),\\
\Phi(x) = \frac{1}{2}(1+ \text{erf}(x/\sqrt{2})),
\end{align}
and $\text{erf}(x)$ is the error function. 

Now let $\psi_{A_1A_2}=\ket{\alpha}\bra{\alpha}_{A_1}\otimes  \ket{0} \bra{0}_{A_2}$ be an input state, where $\ket{\alpha}$ is a coherent state. From \eqref{eq:bs-approx2} it follows that 
\begin{align}\label{eq:fid-bs-bsapprox-uniform}
&F(\B^{\theta}(\psi_{A_1A_2}),\widetilde{\B}^{\theta, \sigma}(\psi_{A_1A_2}) )\nonumber \\&= \int_0^{2\pi} d\theta'~p(\theta', \theta, \sigma, 0, 2\pi) F(\B^{\theta}(\psi_{A_1A_2}),\B^{\theta'}(\psi_{A_1A_2}) ). 
\end{align}

Now consider that  
\begin{align}
F(\B^{\theta}(\psi_{A_1A_2}),&\B^{\theta'}(\psi_{A_1A_2}) ) \nonumber\\
&= F(\psi_{A_1A_2}, \B^{(-\theta)}\circ \B^{\theta'}(\psi_{A_1A_2}))\\
& = \exp[-2\vert \alpha \vert^2(\sin(\theta'-\theta))^2], 
\end{align}
which converges to zero as $\vert \alpha \vert^2\to \infty$. Then from \eqref{eq:fid-bs-bsapprox-uniform} and by an application of the dominated convergence theorem, it follows that the sequence $\{\widetilde{\B}^{\theta, \sigma} \}_{\sigma\in [0, \infty)}$ does not converge uniformly to the ideal beamsplitter $\B^{\theta}$.

We now show that convergence occurs in the strong sense. Let $\psi_{RA_1A_2}$ denote the input state. Then the following holds:
\begin{multline}\label{eq:fid-bs2-approx}
F(\B^{\theta}(\psi_{RA_1A_2}), \widetilde{\B}^{\theta, \sigma}(\psi_{RA_1A_2}) )  \\
 = \int_{0}^{2\pi}d\theta' p(\theta',\theta, \sigma, 0, 2\pi) F(\B^{\theta}(\psi_{RA_1A_2}), \B^{\theta'}(\psi_{RA_1A_2}))~.
\end{multline}
Since $\lim_{\sigma \to 0} p(\theta',\theta, \sigma) = \delta(\theta - \theta')$, it follows that
\begin{equation}\label{eq:bs2-approx-strong-conv}
\lim_{\sigma \to 0} F(\B^{\theta}(\psi_{RA_1A_2}), \widetilde{\B}^{\theta, \sigma}(\psi_{RA_1A_2}) ) = 1,     
\end{equation}
which implies that the sequence $\{ \widetilde{\B}^{\theta, \sigma} \}_{\sigma \in [0, \infty)}$ converges strongly to $\B^{\theta}$.

We now investigate convergence of the sequence  $\{\widetilde{\B}^{\theta, \sigma}\}_{\sigma\in [0, \infty)}$ to $\B^{\theta}$ in terms of the experimental parameters when there is a finite energy constraint on the input states. Consider the following chain of inequalities:
\begin{align}
\frac{1}{2}&\Vert \B^{\theta} - \widetilde{\B}^{\theta, \sigma} \Vert_{\diamond E}\nonumber\\
 &\leq \frac{1}{2}\int_{0}^{2\pi} d\theta'\ p(\theta', \theta,\sigma, 0, 2\pi) \Vert \B^{\theta} - \B^{\theta'}\Vert_{\diamond E}\\
& = \frac{1}{2}\int_{0}^{2\pi} d\theta'\ p(\theta',\theta,\sigma, 0, 2\pi) \Vert \I -  \B^{\theta'-\theta} \Vert_{\diamond E}\\
& \leq\int_{0}^{2\pi} d\theta'\ p(\theta',\theta,\sigma, 0, 2\pi) 2 \sqrt{E\vert\theta'-\theta\vert}~.\label{eq:bound-bs-approx2}
\end{align}
The first inequality follows from convexity of the trace distance. The first equality follows the unitary invariance of the trace distance. The last inequality follows from \cite[Proposition~3.2]{BD18}.

We denote the upper bound on the Shirokov--Winter energy-constrained diamond distance between $ \B^{\theta}$ and  $\widetilde{\B}^{\theta, \sigma} $, obtained in \eqref{eq:bound-bs-approx2} by $g(\theta, \sigma, E)$:
\begin{equation}\label{eq:bound-bs-approx2-g}
g(\theta, \sigma, E) = \int_{0}^{2\pi} d\theta' p(\theta',\theta,\sigma, 0, 2\pi) 2 \sqrt{E\vert\theta'-\theta\vert}.
\end{equation}

\begin{figure}[ptb]
	\includegraphics[
	width=3.4in
	]{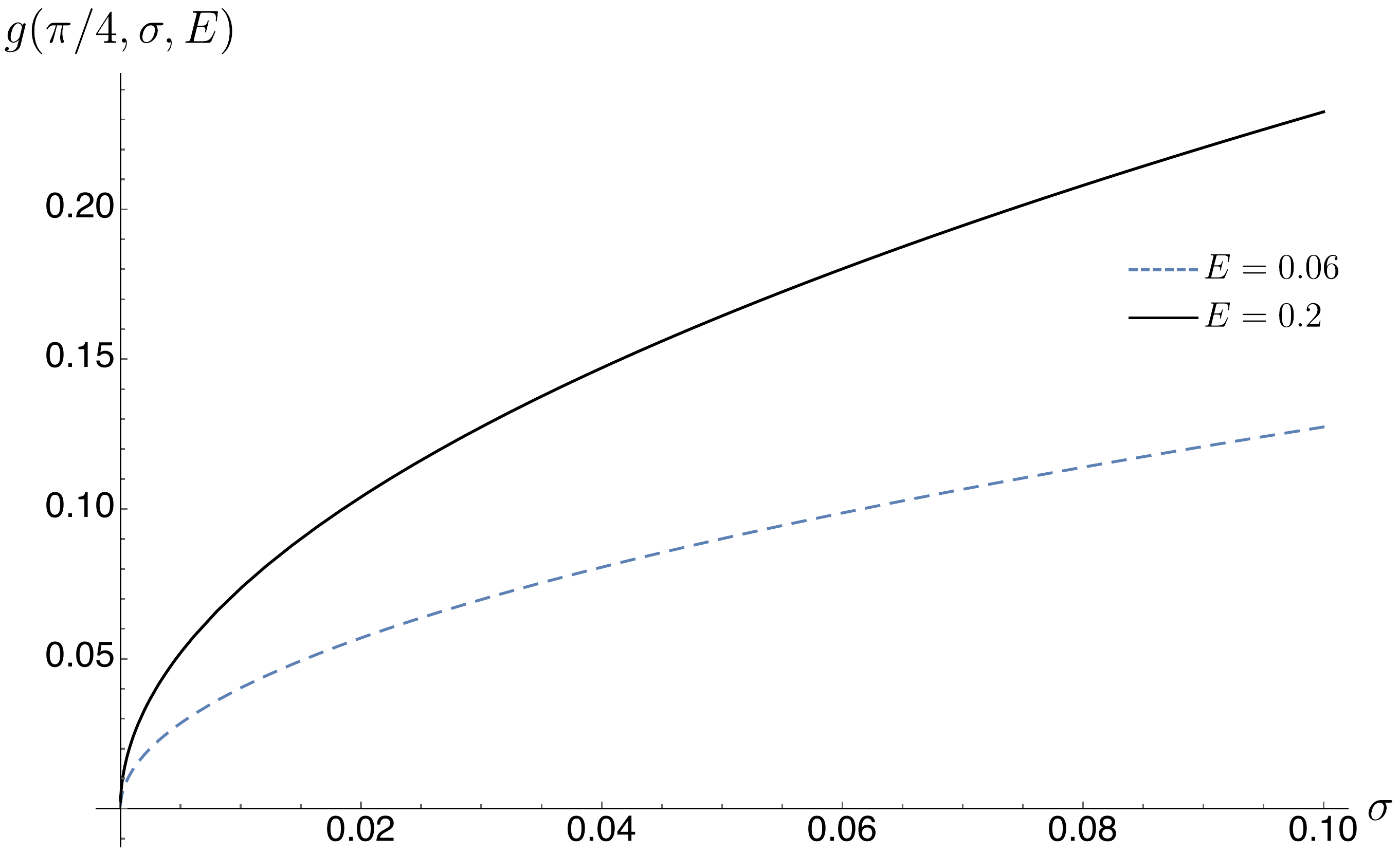}
	\caption{The figure depicts the upper bound $g(\theta, \sigma, E)$ in  \eqref{eq:bound-bs-approx2-g}  for a fixed value of the beamsplitter parameter $\theta = \pi/4$ and for two different values of the energy-constraint parameter  $E=0.06$ and $E=0.2$. Here, $g(\theta, \sigma, E)$ is an upper bound on the Shirokov--Winter energy-constrained diamond distance between an ideal beamsplitter and its experimental approximation in \eqref{eq:bs-approx2}. The figure indicates that  for a fixed value of $\sigma$ in \eqref{eq:truncated-uni-prob},  high accuracy in simulating $\B^{\theta}$ can be achieved only for low values of the energy-constraint parameter $E$.
		\label{fig:bs-bound}
}\end{figure}

In Figure \ref{fig:bs-bound}, we plot $g(\theta, \sigma, E)$ versus $\sigma$ for certain values of the energy constraint $E$ and for $\theta = \pi/4$, which corresponds to the simulation of a balanced beamsplitter. In particular, we find that for all values of $E$, the experimental approximation $\widetilde{\B}^{\theta, \sigma}$ simulates the ideal beamsplitter $\B^{\theta}$ with a high accuracy for $\sigma \approx 0$. Moreover, for a fixed value of $\sigma$, the simulation of $\B^{\theta}$ is more accurate for low values of the energy constraint on input states. 

\section{Approximation of a phase rotation}\label{sec:pr-approx}

In this section, we analyze convergence of the experimental implementation of a phase rotation. In general, the unitary operator corresponding to the phase rotation is given by 
\begin{equation}
    U^{\phi}_{\text{PR}} = \exp(i \hat{n} \phi), 
\end{equation}
where $\hat{n}$ denotes the number operator. 

Similar to Section \ref{sec:bs-approx}, there are at least two ways to model the noise in implementing the unitary channel $\U^{\phi}_{\text{PR}}(\cdot)\equiv U^{\phi}_{\text{PR}}(\cdot) (U^{\phi}_{\text{PR}})^\dag$. A first model consists of sending the input state through a pure-loss channel, followed by the ideal phase rotation. This models the case when some photons are lost in the medium used to implement the phase rotation. We denote the channel corresponding to such an approximation of the ideal phase rotation by $\widetilde{\U}^{\phi, \eta}_{\text{PR}}$, such that
\begin{equation}
\widetilde{\U}^{\phi,\eta}_{\text{PR}}(\rho) \equiv (\U^{\phi}_{\text{PR}}\circ \L^{\eta})(\rho)~.
\end{equation}
Now let $\psi_{RA}$ denote an arbitrary two-mode pure state. Then from unitary invariance of fidelity, we find that 
\begin{equation}
    F(\U^{\phi}_{\text{PR}}(\psi_{RA}), \widetilde{\U}^{\phi, \eta}_{\text{PR}}(\psi_{RA})) =  F(\psi_{RA}, \L^{\eta}_A(\psi_{RA})),
\end{equation}
where it is implicit that the identity channel acts on the reference system $R$. Therefore, analyzing the convergence of the sequence $\{\widetilde{\U}^{\phi,\eta}_{\text{PR}}\}_{\eta\in[0,1) }$ to $\U^{\phi}_{\text{PR}}$ is equivalent to analyzing the convergence of a sequence of pure-loss channels to an identity channel. We note that the same result holds for an ideal displacement unitary and its experimental approximation, as shown in \eqref{eq:displacement-fidelity-simplified}. Therefore, from the results in Section \ref{sec:displacement-approx}, it follows directly that the sequence $\{\widetilde{\U}^{\phi,\eta}_{\text{PR}}\}_{\eta\in[0,1) }$ does not converge  uniformly  to $\U^{\phi}_{\text{PR}}$. Rather the convergence holds in the strong sense. Moreover, the dependence of an estimate of the  Shirokov--Winter energy-constrained diamond distance between $\widetilde{\U}^{\phi,\eta}_{\text{PR}}$ and $\U^{\phi}_{\text{PR}}$ is given by \eqref{eq:buresd-displacement}. From the analysis in Section \ref{sec:SW-energy-constrained-dis-op}, we conclude that only for low values of the energy constraint $E$ on input states and high values of  $\eta$, i.e., low values of the loss, high accuracy in simulating $\U^{\phi}_{\text{PR}}$ can be achieved.  

We now consider a phenomenological model to approximate the ideal phase rotation $\U^{\phi}_{\text{PR}}$. In particular, instead of $\U^{\phi}_{\text{PR}}$, the following channel is applied:
\begin{equation}
    \widetilde{\U}^{\phi, \sigma}\equiv \int_{0}^{2\pi} d\phi'~p(\phi', \phi, \sigma, 0, 2\pi) \U^{\phi'}, 
\end{equation}
where $p(\phi', \phi, \sigma, 0, 2\pi)$ is a truncated normal distribution with location parameter $\phi$ and scale parameter $\sigma$, as defined in \eqref{eq:truncated-uni-prob}. 

Let $\ket{\alpha}$ be an input coherent state. Then from unitary invariance of fidelity it follows that 
\begin{align}
F(\U^{\phi}(\ket{\alpha}\bra{\alpha})&, \U^{\phi'}(\ket{\alpha}\bra{\alpha}) ) \nonumber\\
&= F(\ket{\alpha}\bra{\alpha}, \U^{\phi'-\phi}(\ket{\alpha}\bra{\alpha}))    \\
& = \exp(-2 \vert \alpha \vert^2 (\sin(\phi'-\phi))^2),
\end{align}
which converges to zero as $\vert\alpha \vert^2\to \infty$. Therefore, by an application of the dominated convergence theorem, it follows that the sequence $\{\widetilde{\U}^{\phi,\sigma}\}_{\sigma \in [0,\infty)}$ does not converge uniformly to the ideal phase rotation $\U^{\phi}$.

The strong convergence of $\{\widetilde{\U}^{\phi,\sigma}\}_{\sigma \in [0,\infty)}$ to $\U^{\phi}$ follows from arguments similar those  given in \eqref{eq:fid-bs2-approx} and \eqref{eq:bs2-approx-strong-conv}.

We now provide an estimate of the Shirokov--Winter energy-constrained diamond distance between $\widetilde{\U}^{\phi,\sigma}$ and $\U^{\phi}$. Consider the following chain of inequalities:
\begin{align}
 &\frac{1}{2}\Vert \U^{\phi} - \widetilde{\U}^{\phi, \sigma}\Vert_{\diamond E}  \nonumber \\
 &\leq \frac{1}{2}\int_{0}^{2\pi} d\phi'~p(\phi', \phi, \sigma, 0, 2\pi) \Vert \I - \U^{\phi'- \phi}\Vert_{\diamond E}\\
 & \leq \int_{0}^{2\pi} d\phi'~p(\phi', \phi, \sigma, 0, 2\pi) 2\sqrt{E\vert\phi' - \phi \vert}~.\label{eq:bound-pr-approx2}
\end{align}
The first inequality follows from convexity and unitary invariance of the trace distance. The last inequality follows from \cite[Proposition~3.2]{BD18}. Since the upper bound in \eqref{eq:bound-pr-approx2} is exactly same as the upper bound in \eqref{eq:bound-bs-approx2}, we conclude that only for low values of both the energy constraint $E$ and the scale parameter $\sigma$ in \eqref{eq:truncated-uni-prob}, high accuracy in simulating $\U^{\phi}_{\text{PR}}$ using $\widetilde{\U}^{\phi, \sigma}$ can be achieved.

\section{Approximation of a single-mode squeezer}\label{sec:sms}

In this section, we analyze the convergence of the experimental implementation of a measurement-induced single-mode squeezer from \cite{FMA05} to the ideal single-mode squeezer. 
A single-mode squeezer is a unitary operator defined as
\begin{align}
\label{eq:squeezer}
S(\xi) \equiv \exp[(\xi^{*} \hat{a}^2 - \xi \hat{a}^{\dagger 2})/2]~, 
\end{align}
where $\xi = re^{i \theta}$, with $r \in [0, \infty)$ and $\theta \in [0, 2 \pi]$ (see, e.g., \cite{L15} for a review). A squeezing transformation realizes a decrement in the variance of one of the quadratures at the expense of a corresponding increment in the variance of the complementary quadrature, which is helpful for improving the sensitivity of an interferometer \cite{C81} and for other quantum metrological tasks~\cite{L15}. 

Let $\rho_A$ be an input quantum state, and let $\hat{x}_A$ and $\hat{p}_A$ denote the position- and momentum-quadrature operators for mode $A$, respectively. As described in Figure \ref{fig:SMS-sim}, the simulation from \cite{FMA05} of $\S^r(\rho_A) = S(r)\rho_A S(-r)$, such that $e^{-r} = \sqrt{\eta}$, is given by the following transformation of the mode operators:
\begin{align}
\hat{x}_A & \to \sqrt{\eta} \hat{x}_A + \sqrt{1-\eta}~e^{-r_E} \hat{x}^{0}_E~,  \\
\hat{p}_A & \to \frac{1}{\sqrt{\eta}} \hat{p}_A~,  
\end{align}
where $\hat{x}^0_E$ is the position-quadrature operator corresponding to the vacuum state and $r_E$ is the squeezing parameter corresponding to the squeezed vacuum state. We denote the channel corresponding to the experimental implementation of an ideal single-mode squeezer by $\widetilde{\mathcal{S}}^{\eta, r_E} = \widetilde{\mathcal{S}}^{e^{-2r}, r_E}$. Furthermore, by applying the inverse  $\S^{-r}$ of the ideal single-mode squeezer $\S^{r}$ on the output of $\widetilde{\S}^{e^{-2r}, r_E}$, we arrive at the following transformation:
\begin{align}
\hat{x}_{\text{out}} & = \hat{x}_A + \frac{\sqrt{1-\eta}}{\sqrt{\eta}}~e^{-r_E} \hat{x}^{0}_E~, \label{eq:inverted-squeeze-channel-1} \\
\hat{p}_{\text{out}} & =  \hat{p}_A~.  
\label{eq:inverted-squeeze-channel-2}
\end{align}
We denote the channel induced by the transformation in \eqref{eq:inverted-squeeze-channel-1}--\eqref{eq:inverted-squeeze-channel-2} by $\Xi^{\eta, r_E}$. Since all the elements involved in the transformation are Gaussian, the channel $\Xi^{\eta, r_E}$ can be described by its action on the mean and covariance matrix of the input state $\rho_A$. In particular, there are two $2 \times 2$ real matrices, the scaling matrix $X_{\Xi^{\eta, r_E}} $ and the noise matrix $Y_{\Xi^{\eta, r_E}}$, which characterize the Gaussian channel $\Xi^{\eta, r_E}$ completely (background on Gaussian channels can be found in the appendices). It is easy to check that the action of  $\Xi^{\eta, r_E}$ does not change the mean vector of $\rho_A$. Therefore, the scaling matrix $X_{\Xi^{\eta, r_E}} = I_{2}$, where $I_{2}$ is a two-dimensional identity matrix. Moreover,  the expectation value of the anticommutator $\{\hat{x}_A, \hat{x}^0_E\}$ is equal to zero, which further implies that the noise matrix $Y_{\Xi^{\eta, r_E}}$ has the following form:
$
Y_{\Xi^{\eta, r_E}} =\text{diag}\left(
(1-\eta)e^{-2r_E}/\eta , 0\right)
.$

Let us study the channel $\Xi^{\eta, r_E}$ in further detail. As observed in \cite{ASH07}, all single-mode bosonic Gaussian channels can be categorized into six different canonical forms. In particular, the canonical form $\Phi_{B_1}$ has the following $X_{\Phi_{B_1}}$ and $Y_{\Phi_{B_1}}$ matrices \cite{ASH07}:
\begin{equation}
X_{\Phi_{B_1}} = I_2, \qquad Y_{\Phi_{B_1}} = \text{diag}(0, 1).
\end{equation}
We now show that the channel $\Xi^{\eta, r_E}$ is unitarily equivalent to the canonical form $\Phi_{B_1}$ \cite{ASH07}. Let $\rho$ be a quantum state with the covariance matrix $V_{\rho}$. Then the symplectic matrix $\sigma_x = \begin{bmatrix}
0 & 1\\
1 & 0
\end{bmatrix}$ transforms the covariance matrix $V_{\rho}$ as follows: $V^{\prime}_{\rho}$ = $\sigma_x V_{\rho} \sigma_x$. We then apply the symplectic transformation corresponding to the symplectic matrix $\K =\text{diag}(\varsigma, 1/\varsigma)$, where $\varsigma = \sqrt{(1-\eta)}e^{-r_E}/\sqrt{\eta}$,  on the covariance matrix $V^{\prime}_{\rho}$. The transformed covariance matrix is given by $V^{\prime \prime}_{\rho} = \K V^{\prime}_{\rho} \K$. We now apply the canonical form $\Phi_{B_1}$ on the transformed state, and get the following transformation of the covariance matrix $V^{\prime \prime}_{\rho}$: $V^{\prime \prime \prime}_{\rho} = V^{\prime \prime}_{\rho} + Y_{\Phi_{B_1}}$. We then apply the symplectic transformation corresponding to the symplectic matrix $\K^{-1}$ followed by $\sigma_x$ on $V^{\prime \prime \prime}_{\rho}$, and get the following final covariance matrix $V^{\text{final}}_{\rho}$:
\begin{align}
V^{\text{final}}_{\rho} &= \sigma_x \K^{-1}V^{\prime  \prime}_{\rho} \K^{-1} \sigma_x+ \sigma_x \K^{-1} Y_{\Phi_{B_1}}\K^{-1} \sigma_x  \\
	&=V_{\rho} + \text{diag}(\varsigma^2,0)  \\
	&= V_{\rho} + \text{diag}\left(
	(1-\eta)e^{-2r_E}/\eta , 0\right)~,
	\end{align}
which implies that the overall transformation is the same as the action of the channel $\Xi^{\eta, r_E}$ on the state $\rho$. Therefore, we have shown that the Gaussian channel $\Xi^{\eta, r_E}$ is unitarily equivalent by Gaussian input and output unitaries to the canonical form $\Phi_{B_1}$. This gives a physical interpretation to channels in the class $\Phi_{B_1}$, in terms of the measurement-induced squeezing approximation from \cite{FMA05}.

\begin{figure}[ptb]
	\begin{center}
		\includegraphics[
		width=3.3in
		]{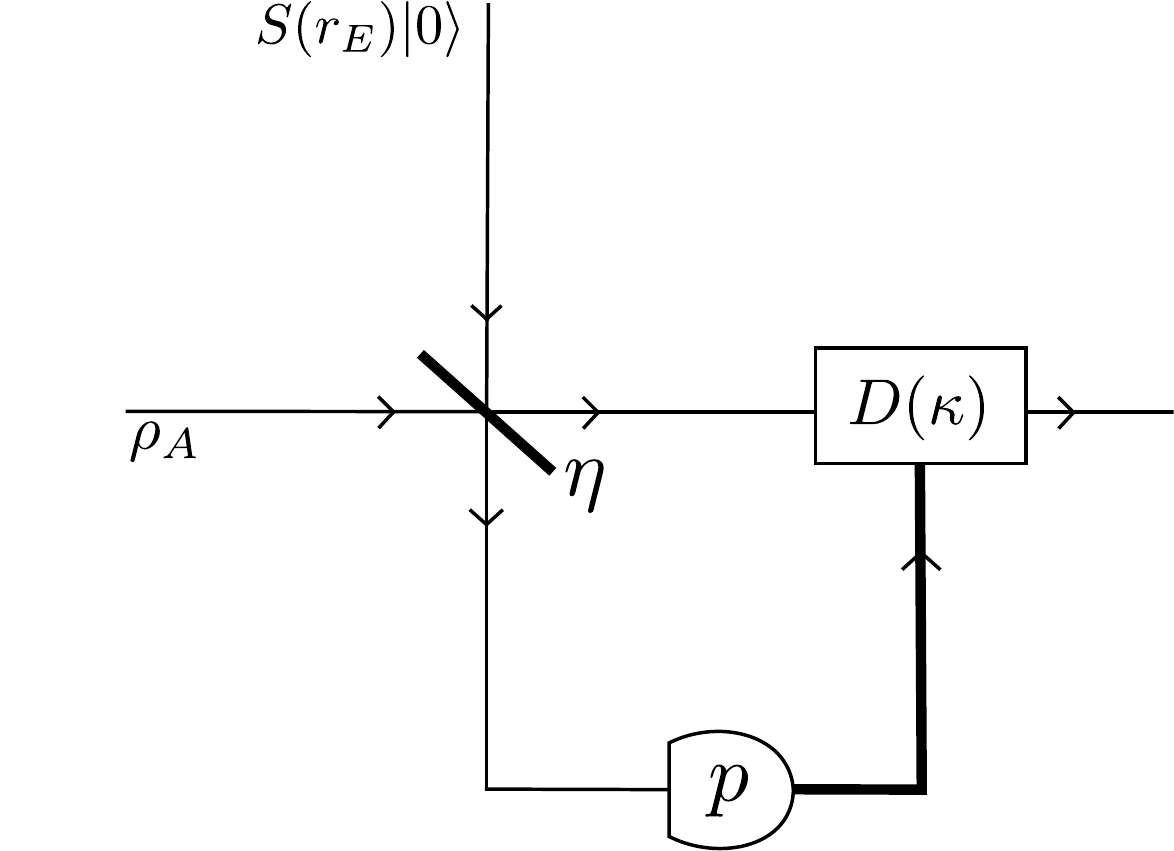}
	\end{center}
	\caption{The figure plots an experimental approximation of the ideal single-mode squeezing unitary $S(r)$ on the input state $\rho_A$, such that $e^{-r} = \sqrt{\eta}$.  $S(r_E)\ket{0}_E$ represents a squeezed vacuum state in the mode $E$. The experimental approximation of an ideal single-mode squeezing operation corresponds to the following transformations: sending $\rho_A$ and $S(r_E)\ket{0}_E$ through a beamsplitter with transmissivity $\eta$ followed by a measurement of the momentum quadrature in mode $E$. Then a feed forward operation corresponding to the measurement outcome $p$ followed by a displacement operator $D(\kappa)$ on mode $A$ \cite{FMA05}, where $\kappa = -i \sqrt{(1-\eta)/(2\eta)}p$. }
	\label{fig:SMS-sim}\end{figure}

\subsection{Lack of uniform convergence}

We now prove that the sequence $\{ \widetilde{\S}^{e^{-2r}, r_E}  \}_{r_E \in [0, \infty)}$ does not converge uniformly to the ideal single-mode squeezer $\S^r$. Let $\ket{z}$ be a squeezed-vacuum input state with the covariance matrix $V_{\ket{z}\bra{z}} = \text{diag}(z, 1/z)$ and mean vector $\mu_{\ket{z}\bra{z}} = (0,0)^{\text{T}}$, where $z \in [0,\infty)$. Then under the action of $\Xi^{\eta, r_E}$, the covariance matrix $V_{\ket{z}\bra{z}}$ transforms as follows:
\begin{equation}
V^{\text{out}}_{\ket{z}\bra{z}} = \text{diag}\big(
z+ (1-\eta)e^{-2r_E}/\eta, 0\big)
.
\end{equation}
We can use these expressions in the Uhlmann fidelity formula for single-mode Gaussian states \cite{HS98, PHS98}. By using the unitary invariance of fidelity, we find that 
\begin{multline}
F(\S^r(\ket{z}\bra{z}), \widetilde{\S}^{e^{-2r}, r_E}(\ket{z}\bra{z})) \\
= F(\ket{z}\bra{z}, \Xi^{\eta, r_E}(\ket{z}\bra{z})). 
\end{multline}
Moreover, we find that \cite{Mathematica}
\begin{align}
F(\ket{z}\bra{z}, \Xi^{\eta, r_E}(\ket{z}\bra{z})) = \sqrt{\frac{2z}{2z+(e^{2r}-1)e^{-2r_E}}}~.
\end{align}
Therefore, for fixed $r_E$
\begin{equation} \label{eq:lack-of-uniform-conv-for-SMS}
\lim_{z \to 0}F(\S^r(\ket{z}\bra{z}), \widetilde{\S}^{e^{-2r}, r_E}(\ket{z}\bra{z})) = 0~,
\end{equation}
which implies that the sequence $\{ \widetilde{\S}^{e^{-2r}, r_E}  \}_{r_E \in [0, \infty)}$ does not converge uniformly to the ideal single-mode squeezer transformation $\S^r$.

The reasoning behind \eqref{eq:lack-of-uniform-conv-for-SMS} can be intuitively explained as follows: the channel $\Xi^{\eta, r_E}$ adds noise to the $\hat{x}$ quadrature only. Therefore, it can be discriminated from an identity channel by using an input state that has vanishing noise in the $\hat{x}$ quadrature operator. Since an infinitely squeezed vacuum state (infinitely squeezed in the position quadrature) satisfies such a condition, then \eqref{eq:lack-of-uniform-conv-for-SMS} follows. 

\subsection{Strong convergence}

We now argue that the sequence $\{\widetilde{\S}^{e^{-2r}, r_E}\}_{r_E \in [0,\infty)}$ converges to $\S^r$ in the strong sense. Let $\chi_{\rho_A}(x, p)$ denote the Wigner characteristic function of the input state~$\rho_A$. Let $\tilde{\rho}^{\text{out}}_A$ denote the state after the action of $\widetilde{\S}^{e^{-2r}, r_E}$ on $\rho_A$: $\tilde{\rho}_A^{\text{out}} = \widetilde{\S}^{e^{-2r}, r_E}(\rho_A)$. Then the characteristic function of $\tilde{\rho}^{\text{out}}_A$ is given by  
\begin{align}
\chi_{\tilde{\rho}^{\text{out}}_A}(x, p) = \chi_{\rho_A}(e^rx, e^{-r}p) e^{-\frac{1}{4}(e^{2r}-1)e^{-2r_E} p^2}~. 
\end{align}
Moreover, the characteristic function of $\S^r(\rho_A)$ is given by 
\begin{align}
\chi_{\S^r(\rho_A)}(x, p) = \chi_{\rho}(e^r x, e^{-r}p)~. 
\end{align}
Therefore, for each $\rho_A \in \D(\H_A)$, and  for all $x, p \in \mathbb{R}$
\begin{align}\label{eq:strong-conv-SMS}
\lim_{r_E \to \infty}  \chi_{\tilde{\rho}^{\text{out}}_A}(x, p) = \chi_{\S^r(\rho_A)}(x, p)~.
\end{align}
Therefore, we have shown that the sequence of characteristic functions $\chi_{\tilde{\rho}^{\text{out}}_A}(x, p) $ converges pointwise to $\chi_{\S^r(\rho_A)}(x, p)$, which implies that the sequence $\{\widetilde{\S}^{e^{-2r}, r_E}\}_{r_E \in [0,\infty)}$ converges strongly to $\S^r$ \cite[Lemma~8]{LSW18}. 

As described in Figure \ref{fig:SMS-sim}, the simulation of an ideal single-mode unitary consists of an ideal displacement. We now briefly discuss the case when the displacement operator involved in the simulation of $\S^r$ is not ideal. By using the counterexample from before, we find that convergence of the simulation of a single-mode squeezing operation to an ideal single-mode squeezing operation is not uniform. From \eqref{eq:strong-conv-displacement}, \eqref{eq:strong-conv-SMS}, and \cite[Proposition~2]{M18}, it follows that convergence holds in the strong sense. 

Furthermore, the strong convergence of the sequence $\{\widetilde{\S}^{e^{-2r}, r_E}\}_{r_E\in [0, \infty)}$ to $\S^r$ implies that the experimental approximations of an ideal single-mode squeezer, as described in Figure \ref{fig:SMS-sim}, simulate the desired unitary operation uniformly on the set density operators whose marginals on the channel input have bounded energy~\cite{S18}. However, as discussed previously, from an experimental perspective, it is important to know how this convergence depends on the experimental parameters. We now consider experimentally relevant input Gaussian states with energy constraints, such as single-mode squeezed states, coherent states, and two-mode squeezed vacuum states. For any fixed finite value of the energy constraint, we find that,  among these Gaussian states, inputting a two-mode squeezed vacuum state provides the largest value of the sine distance between the ideal single-mode squeezer and its experimental approximation. 

\subsection{Estimates of the Shirokov--Winter energy-constrained diamond norm}

Let us study  in detail the case when the input state is the two-mode squeezed vacuum state with parameter~$N$, as defined in \eqref{eq:tms}. The fidelity between $\S^r(\psi_{\operatorname{TMS}}(N))$ and $\widetilde{\S}^{e^{-2r}, r_E}(\psi_{\operatorname{TMS}}(N))$ is given by \cite{Mathematica}
\begin{multline}\label{eq:fidelity-squeezing-tms}
F(\S^r(\psi_{\operatorname{TMS}}(N)), \widetilde{\S}^{e^{-2r}, r_E}(\psi_{\operatorname{TMS}}(N)))  \\
= \frac{1}{\sqrt{1+(N+1/2)(e^{2r}-1 )e^{-2r_E} }}~.
\end{multline}

Next, we perform numerical evaluations to see how close the experimental approximation $\widetilde{\S}^{e^{-2r}, r_E}$ is to the ideal squeezing operation $\S^r$ for a fixed input quantum state $\psi_{\text{TMS}}(N)$. Fix the squeezing parameter $r= 0.46$, which corresponds to the squeezing strength 4~dB.  We use the relation $10 \log_{10}(\exp(2r)) \approx 8.686 r$ to convert the squeezing parameter $r$ to units of dB. Let $g(r_E, N)$ denote the sine distance between $\S^r(\psi_{\operatorname{TMS}}(N))$ and~$\widetilde{\S}^{e^{-2r}, r_E}(\psi_{\operatorname{TMS}}(N))$:
\begin{align}\label{eq:sd-squeezing}
g(r_E, N) = \sqrt{1- \frac{1}{\sqrt{1+(N+1/2)(e^{0.92}-1 )e^{-2r_E} }}}~,
\end{align}
where we used \eqref{eq:fidelity-squeezing-tms}.

In Figure~\ref{fig:energy-constraind-squeezing}, we plot $g(r_E, N)$ in \eqref{eq:sd-squeezing} versus the offline squeezing strength $r_E$ for certain values of the input mean photon number $N$. In particular, we find that the simulation of $\S^{r}$ is more accurate for low values of the energy constraint on the input states. The figure indicates that an offline squeezing strength of 15 dB, which is what is currently experimentally achievable \cite{VMDS16}, is not sufficient to simulate an ideal squeezing operation with squeezing strength 4 dB, with a high accuracy, by using the measurement-induced protocol from \cite{FMA05}.

\begin{figure}[ptb]
	\includegraphics[
	width=3.3in
	]{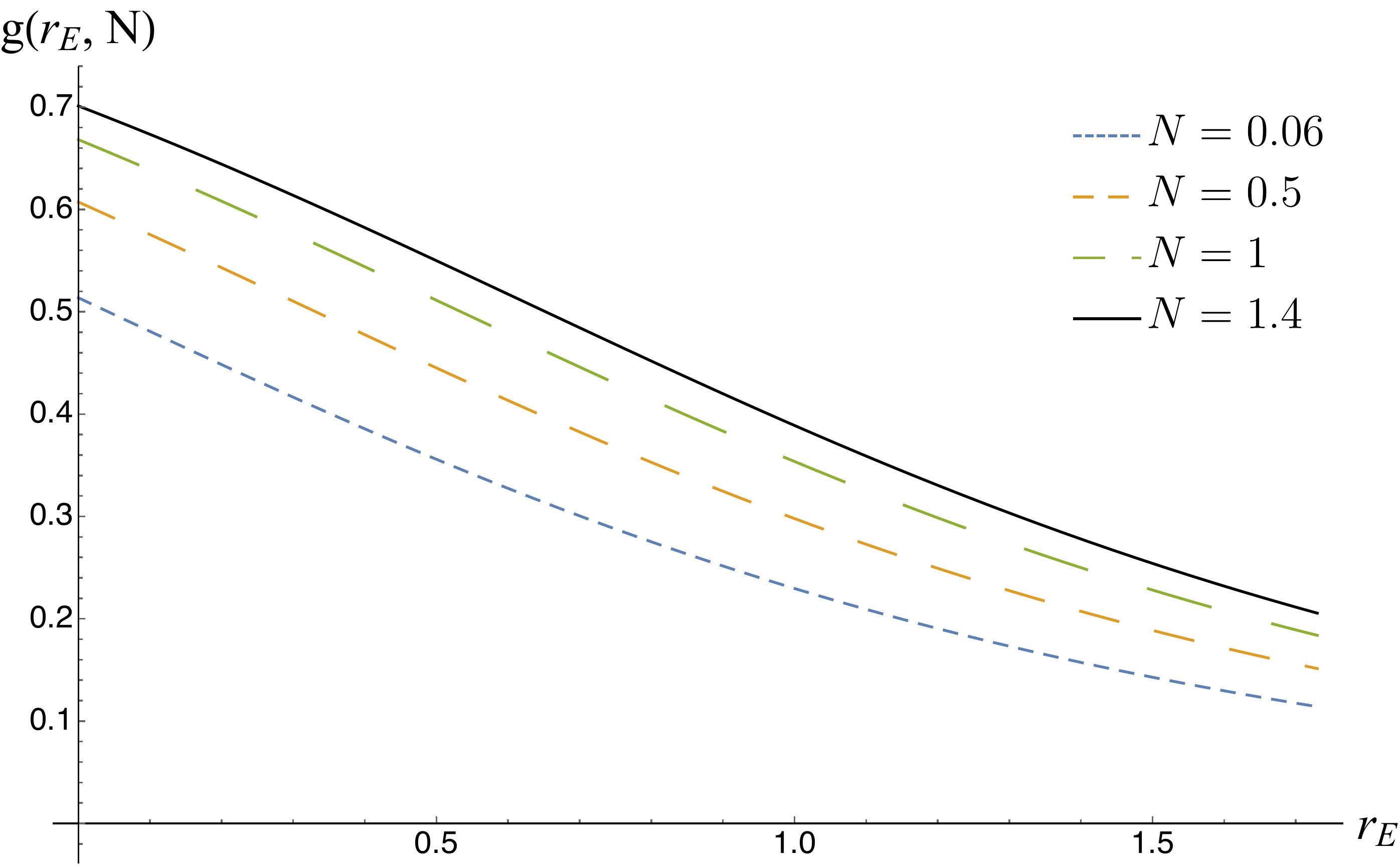}
	\caption{The figure plots the sine distance $g(r_E, N)$ in \eqref{eq:sd-squeezing} between an ideal single-mode squeezer $\S^{r}$ with 4~dB squeezing strength and its experimental approximation $\widetilde{S}^{e^{-2r}, r_E}$ when the input state is the two-mode squeezed vacuum state with parameter $N$, as defined in \eqref{eq:tms}. In the figure, we select certain values of the mean-photon number $N$ of the channel input, with the choices indicated next to the figure. For a fixed value of $r_E$, the simulation of $\S^{r}$ is more accurate for low values of the energy constraint on input states. The figure indicates that an offline squeezing strength  of 15~dB, which is what is currently experimentally achievable \cite{VMDS16}, is not sufficient to simulate an ideal squeezing operation  with squeezing strength 4~dB, with a high accuracy, by using the protocol from~\cite{FMA05}.}
	\label{fig:energy-constraind-squeezing}
\end{figure}

We further investigate the strength of the offline  squeezing  required to simulate the ideal squeezing operator with high accuracy. In Figure \ref{fig:energy-constraind-squeezing2}, we plot Figure \ref{fig:energy-constraind-squeezing} for high values of the squeezing parameter $r_E$. The figure indicates that for the low input mean photon number $N \approx 0.06$, approximately 26 dB offline squeezing strength is required to achieve a reasonable accuracy~($\approx 97\%$).

\begin{figure}[ptb]
	\includegraphics[
	width=3.3in
	]{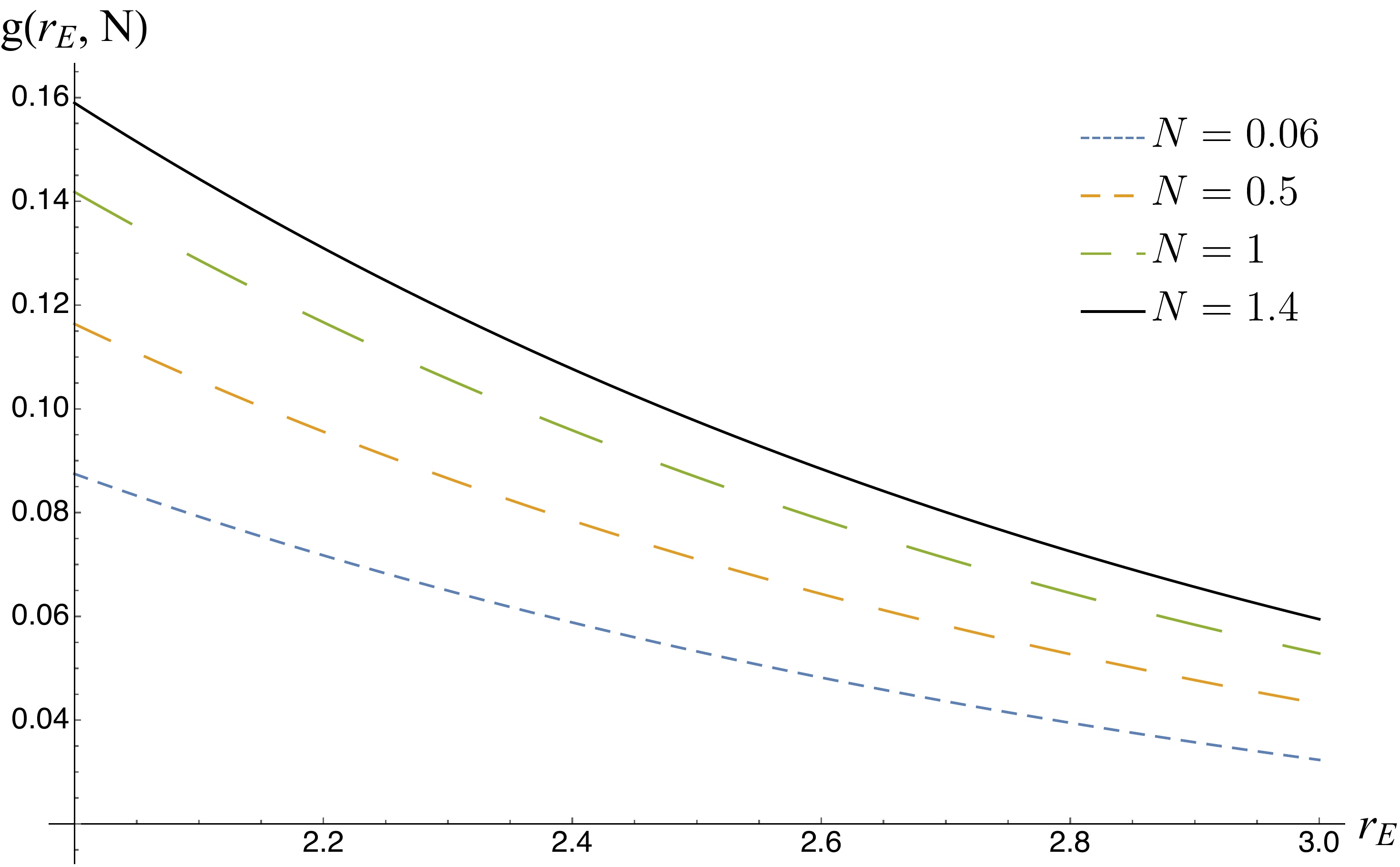}
	\caption{The figure plots Figure \ref{fig:energy-constraind-squeezing} for high values of the offline squeezing parameter $r_E$. The figure indicates that only for high values of the offline squeezing parameter $r_E$ and low values of $N$, high accuracy in simulating $\S^r$ with $4$ dB squeezing strength can be achieved. }
	\label{fig:energy-constraind-squeezing2}
\end{figure}

\section{Approximation of a SUM gate}

In this section, we analyze the convergence of experimental approximations of a measurement-induced $\text{SUM}$ gate from \cite{FMA05} to the ideal $\text{SUM}$ gate. A $\text{SUM}$ gate is a quantum nondemolition (QND) interaction between two modes, the CV analog of the CNOT gate~\cite{BSBN02}:
\begin{equation}
\text{SUM}^G \equiv \exp(- i~ G \hat{x}_1 \otimes \hat{p}_2)~,\label{eq:SUM}
\end{equation} 
where $\hat{x}_1$ and $\hat{p}_2$ correspond to the position- and momentum-quadrature operators of modes 1 and 2, respectively, and $G$ is the gain of the interaction. Generally, $G=1$ is sufficient for quantum information processing tasks. Other than universal quantum computation, this CV entangling quantum gate has applications in CV quantum error correction \cite{SLB1998,LS98} and CV coherent communication \cite{WKB07,PhysRevA.77.022321}.

Let $\rho^{\text{in}}_{12}$ denote a two-mode input quantum state. Then the action of the ideal $\text{SUM}^G$ gate on the mode operators $\hat{x}_1$, $\hat{x}_2$, $\hat{p}_1$, and $\hat{p}_2$ of $\rho^{\text{in}}_{12}$ is given by 
\begin{align}
\hat{x}^{\text{in}}_1 & \to \hat{x}^{\text{in}}_1~,  \\
\hat{p}^{\text{in}}_1 & \to \hat{p}^{\text{in}}_1 - G \hat{p}^{\text{in}}_2~,  \\
\hat{x}^{\text{in}}_2 & \to \hat{x}^{\text{in}}_2 + G \hat{x}^{\text{in}}_1~, \\
\hat{p}^{\text{in}}_2 & \to \hat{p}^{\text{in}}_2~.
\end{align}

On the other hand, as described in Figure \ref{fig:SUM-sim}, the simulation of $\text{SUM}^G(\rho^{\text{in}}_{12})$ from \cite{FMA05} is given by the following transformation of the mode operators $\hat{x}_1$, $\hat{x}_2$, $\hat{p}_1$, and $\hat{p}_2$ of $\rho^{\text{in}}_{12}$:
\begin{align}
\hat{x}^{\text{in}}_1  & \to \hat{x}^{\text{in}}_1 - \sqrt{\frac{1-R}{1+R}} e^{-r_A} \hat{x}^0_A~,  \\
\hat{p}^{\text{in}}_1 & \to \hat{p}^{\text{in}}_1 - G \hat{p}^{\text{in}}_2 + \sqrt{\frac{R(1-R)}{1+R}} e^{-r_B} \hat{p}^0_B~,  \\
\hat{x}^{\text{in}}_2 & \to \hat{x}^{\text{in}}_2 + G \hat{x}^{\text{in}}_1 + \sqrt{\frac{R(1-R)}{1+R}}e^{-r_A}\hat{x}^0_A~, \\
\hat{p}^{\text{in}}_2 & \to \hat{p}^{\text{in}}_2 + \sqrt{\frac{1-R}{1+R}} e^{-r_B} \hat{p}^0_B~,
\end{align}
where $G = 1/\sqrt{R} - \sqrt{R}$, $r_A$ and $r_B$ denote the squeezing parameter corresponding to the modes $A$ and $B$, respectively, and $0 <R \leq 1$. We denote the channel corresponding to the experimental implementation of an ideal $\text{SUM}^G$ by $\widetilde{\text{SUM}}^{r_A, r_B, R}$.  Furthermore, by applying the inverse of $\text{SUM}^G$ on the output of $\widetilde{\text{SUM}}^{r_A, r_B, R}$, we get the following transformation of the mode operators:
\begin{align}
\hat{x}^{\text{out}}_1 & =  \hat{x}^{\text{in}}_1 - \sqrt{\frac{1-R}{1+R}} e^{-r_A} \hat{x}^0_A~, \\
\hat{p}^{\text{out}}_1 & =   \hat{p}^{\text{in}}_1 + \sqrt{\frac{1-R}{R(1+R)}} e^{-r_B} \hat{p}^0_B~, \\
\hat{x}^{\text{out}}_2 & = \hat{x}^{\text{in}}_2 + \sqrt{\frac{1-R}{R(1+R)}}e^{-r_A}\hat{x}^0_A~, \\
\hat{p}^{\text{out}}_2 & =  \hat{p}^{\text{in}}_2 + \sqrt{\frac{1-R}{1+R}} e^{-r_B} \hat{p}^0_B~.
\end{align}
We denote the channel induced by this overall transformation by $\Lambda^{r_A, r_B, R}$. Since all the elements involved in the transformation are Gaussian, the channel $\Lambda^{r_A, r_B, R}$ can be described by its action on the mean vector and covariance matrix of the input state $\rho^{\text{in}}_{12}$. We now find two $4\times4$ real matrices $X_{\Lambda^{r_A, r_B, R}}$ and $Y_{\Lambda^{r_A, r_B, R}}$, which characterize the Gaussian channel $\Lambda^{r_A, r_B, R}$ completely (background on Gaussian channels can be found in the appendices). 
From the aforementioned equations, it is clear that the mean vector of $\rho^{\text{in}}_{12}$ is invariant under the action of the channel $\Lambda^{r_A, r_B, R}$. Therefore, the scaling matrix $X_{\Lambda^{r_A, r_B, R}} = I_4$, where $I_4$ is a four-dimensional identity matrix. Moreover, the noise matrix $Y_{\Lambda^{r_A, r_B, R}}$ has the following form: 
\begin{equation}
 Y_{\Lambda^{r_A, r_B, R}} = 
\begin{bmatrix}
\alpha(r_A)& 0 & -\frac{\alpha(r_A)}{\sqrt{R}} & 0\\
0 & \frac{\beta(r_B)}{R} & 0 & \frac{\beta(r_B)}{\sqrt{R}}\\
 -\frac{\alpha(r_A)}{\sqrt{R}}  & 0 & \frac{\alpha(r_A)}{R} & 0\\
 0 & \frac{\beta(r_B)}{\sqrt{R}} & 0 &  \beta(r_B)~
\end{bmatrix}~,
\end{equation}
where
\begin{align}
\alpha(r_A) & = [(1-R)e^{-2r_A}]/(1+R), \\
\beta(r_B) & = [(1-R)e^{-2r_B}]/(1+R).
\end{align}

\begin{figure}[ptb]
	\begin{center}
		\includegraphics[
		width=3.3in
		]{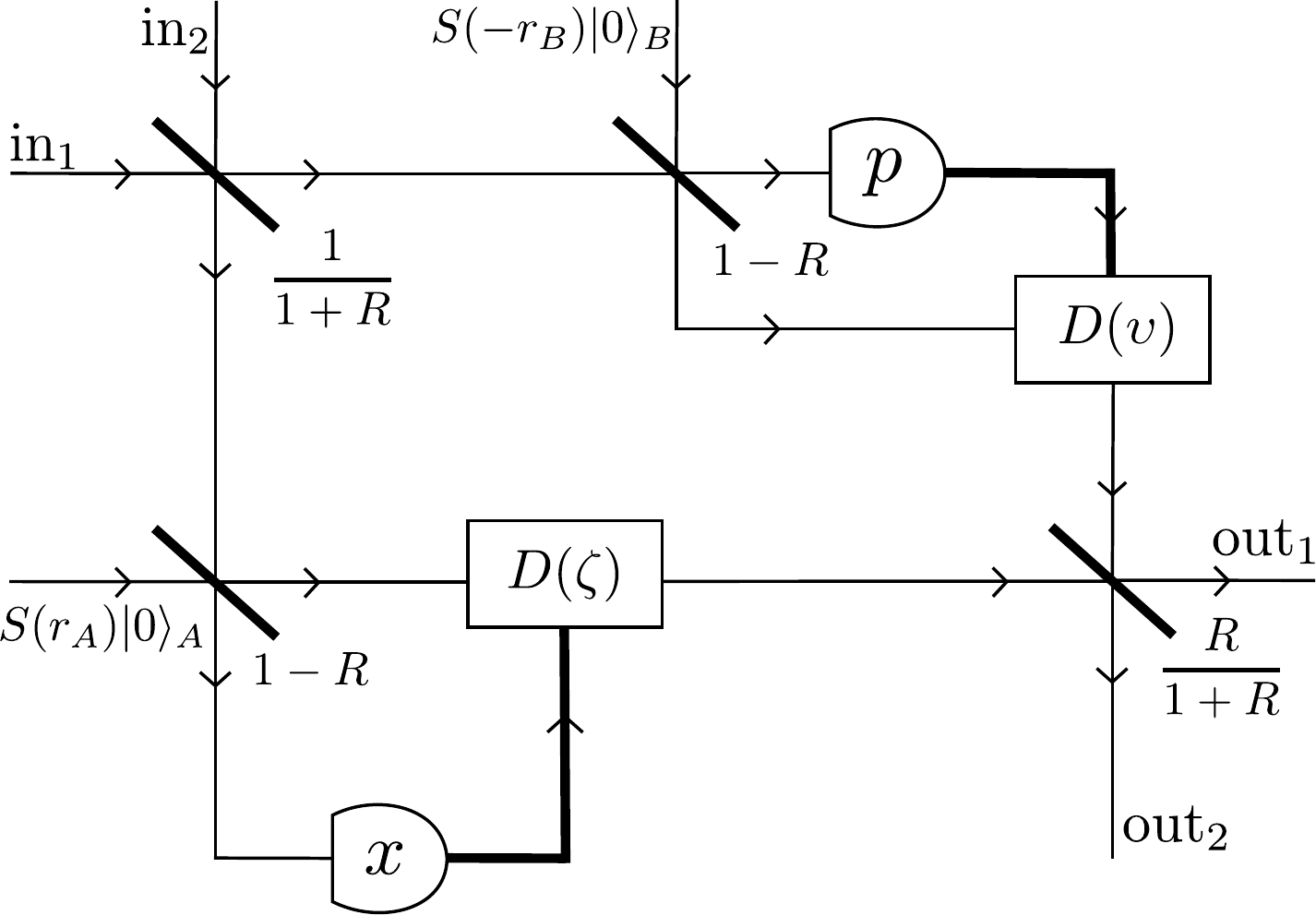}
	\end{center}
	\caption{The figure plots an experimental approximation of the ideal SUM gate ($\text{SUM}^G = \exp(- i~ G \hat{x}_1 \otimes \hat{p}_2)$, where $G = 1/\sqrt{R} - \sqrt{R}$, and $0<R \leq 1$) on a two-mode input quantum state. The circuit consists of a sequence of passive transformations, off-line squeezed vacuum states, homodyne measurements, feed-forward operations, and displacement unitaries \cite{FMA05}. $D(\zeta)$ and $D(\upsilon)$ denote displacement unitaries with $\zeta = -\sqrt{1-R}x/\sqrt{2R}$ and $\upsilon= - \sqrt{1-R}p/\sqrt{2R}$, respectively.}
	\label{fig:SUM-sim}\end{figure}

\subsection{Lack of uniform convergence}

We now prove that the sequence  $\{\widetilde{\text{SUM}}^{r_A, r_B, R}\}_{r_A, r_B \in [0,\infty)}$ does not converge  uniformly to the ideal $\text{SUM}^G$ gate. Let $\ket{\psi}_{12} = \ket{z}_1\ket{z}_2$, where $\ket{z}$ denotes a single-mode squeezed-vacuum state with the covariance matrix $V_{\ket{z}\bra{z}} = \text{diag}(z, 1/z)$, where $z \in [0,\infty)$. The covariance matrix of $\psi_{12}$ is $V_{\psi_{12}} = \text{diag}(z, 1/z, z, 1/z)$, and its mean vector is $\mu_{\psi_{12}} = (0,0,0,0)^{\text{T}}$. Under the action of $\Lambda^{r_A, r_B, R}$, the covariance matrix $V_{\psi_{12}}$ transforms as follows:
\[V^{\text{out}}_{\psi_{12}} = 
\begin{bmatrix}
z+ \alpha(r_A) & 0 &  -\frac{\alpha(r_A)}{\sqrt{R}} & 0 \\
0 & \frac{1}{z}+\frac{ \beta(r_B)}{R} & 0 &  \frac{\beta(r_B)}{\sqrt{R}} \\
 -\frac{\alpha(r_A)}{\sqrt{R}} & 0 & z+\frac{\alpha(r_A)}{R} & 0 \\
0 &  \frac{\beta(r_B)}{\sqrt{R}} & 0 & \frac{1}{z}+\beta(r_B)
\end{bmatrix}.
\]

We now use these expressions in the Uhlmann fidelity formula for two-mode Gaussian states~\cite{MM12}. By using the unitary invariance of fidelity, we find that
\begin{multline}
F(\text{SUM}^G(\psi_{12}), \widetilde{\text{SUM}}^{r_A, r_B, R}(\psi_{12})) \\
= F(\psi_{12}, \Lambda^{r_A, r_B, R}(\psi_{12})).
\end{multline}
Moreover, we find that \cite{Mathematica}  
\begin{multline}
F(\psi_{12}, \Lambda^{r_A, r_B, R}(\psi_{12}))  \\
= \frac{2\sqrt{z}R}{\sqrt{(2z R + (1-R)e^{-2r_A})(2R+z(1-R)e^{-2r_B})}}.
\end{multline}
Therefore, for fixed $r_A, r_B$ 
\begin{align}\label{eq:lack-of-uniform-con-SUM-gate}
\lim_{z\to 0} F(\text{SUM}^G(\psi_{12}), \widetilde{\text{SUM}}^{r_A, r_B, R}(\psi_{12})) = 0~,
\end{align}
which implies that the sequence $\{\widetilde{\text{SUM}}^{r_A, r_B, R}\}_{r_A, r_B \in [0,\infty)}$ does not converge uniformly to the ideal $\text{SUM}^G$ gate.

\subsection{Strong convergence}

We now argue that the sequence  $\{\widetilde{\text{SUM}}^{r_A, r_B, R}\}_{r_A, r_B \in [0,\infty)}$ converges to the $\text{SUM}^G$ gate in the strong sense. Let $\rho^{\text{in}}_{12}$ denote the input state. Let $\chi_{\text{SUM}^G(\rho^{\text{in}}_{12})}(x_1, p_1, x_2, p_2)$ denote the characteristic function of the state $\text{SUM}^G(\rho^{\text{in}}_{12})$.
Let $\tilde{\rho}^{\text{out}}_{12}$ denote the state after the action of $\widetilde{\text{SUM}}^{r_A, r_B, R}$ on $\rho_{12}$: $\tilde{\rho}^{\text{out}}_{12} = \widetilde{\text{SUM}}^{r_A, r_B, R}(\rho_{12})$. 
 Then the characteristic function of $\tilde{\rho}^{\text{out}}_{12}$ is given by 
\begin{multline}
\chi_{\tilde{\rho}^{\text{out}}_{12}}(x_1, p_1, x_2, p_2) = \chi_{\text{SUM}^G(\rho^{\text{in}}_{12})}(x_1, p_1, x_2, p_2)\times  \\
 \exp\!\Bigg(\frac{R-1}{4(1+R)}[(p_1-\sqrt{R}p_2)^2 e^{-2r_A}\\
 +(\sqrt{R}x_1+x_2)^2e^{-2r_B}]\Bigg). 
 \notag
\end{multline}
Therefore, for each $ \rho^{\text{in}}_{12}\in \D(\H_1 \otimes \H_2)$, and for all $x_1, p_1, x_2, p_2 \in \mathbb{R}$
\begin{align}\label{eq:strong-conv-SUM}
\lim_{r_A, r_B \to \infty}\chi_{\tilde{\rho}^{\text{out}}_{12}}(x_1, p_1, x_2, p_2)=  \chi_{\text{SUM}^G(\rho^{\text{in}}_{12})}(x_1, p_1, x_2, p_2),
\end{align}
which implies that $\{\widetilde{\text{SUM}}^{r_A, r_B, R}\}_{r_A, r_B \in [0,\infty)}$ converges strongly to the $\text{SUM}^G$ gate. 

\subsection{Unideal displacements}

As described in Figure \ref{fig:SUM-sim}, the simulation of an ideal SUM gate consists of two ideal displacements. We now briefly discuss the case when these displacement operators are not ideal. From the counterexamples given previously, we find that convergence of the simulation of a SUM gate to an ideal SUM gate is not uniform. By using the triangle inequality for sine distance, \eqref{eq:strong-conv-displacement}, \eqref{eq:strong-conv-SUM}, and \cite[Proposition~2]{M18}, the convergence holds in the strong sense. Moreover, the strong convergence of the sequence $\{\widetilde{\text{SUM}}^{r_A, r_B, R}\}_{r_A, r_B \in [0,\infty)}$  to the $\text{SUM}^G$ gate implies that the experimental approximations of an ideal  $\text{SUM}$ gate simulate the desired unitary operation uniformly on the set of density operators whose marginals on the channel input have bounded energy~\cite{S18}. 

\subsection{Estimates of the Shirokov--Winter energy-constrained diamond norm}

Similar to Section \ref{sec:sms}, we investigate the dependence of the convergence of the sequence $\{\widetilde{\text{SUM}}^{r_A, r_B, R}\}_{r_A, r_B \in [0,\infty)}$  to the $\text{SUM}^G$ gate on the experimental parameters when there is a finite energy constraint on the input states. Since the $\text{SUM}^G$ gate acts on two modes, we consider several experimentally relevant quantum states with energy constraints, such as a tensor product of two coherent states, a tensor product of two single-mode squeezed states, a two-mode squeezed vacuum state, and a tensor product of two two-mode squeezed vacuum states. For a fixed finite value of the energy constraint, we find that a tensor product of two two-mode squeezed vacuum states provides the largest value of the sine distance between the ideal $\text{SUM}$ gate and its experimental approximation. 

We now discuss in detail the case when the input state is a tensor product of two two-mode squeezed vacuum states with parameter $N$, as defined in \eqref{eq:tms}. For $G = 1/\sqrt{R} - \sqrt{R}$, and $0<R \leq 1$, the fidelity between $\text{SUM}^G(\psi^{\otimes 2}_{\operatorname{TMS}}(N))$ and $\widetilde{\text{SUM}}^{r_A, r_B, R}(\psi^{\otimes 2}_{\operatorname{TMS}}(N) )$ is given by \cite{Mathematica}
\begin{align}\label{eq:fid-SUM}
F&\Big(\widetilde{\text{SUM}}^{r_A, r_B, R}(\psi^{\otimes 2}_{\operatorname{TMS}}(N)), \text{SUM}^G(\psi^{\otimes 2}_{\operatorname{TMS}}(N))\big) \nonumber \\
& = \frac{2 R e^{r_A + r_B}}{\sqrt{ (\kappa(N, R) - 2 e^{2r_A} R   )( \kappa(N, R) - 2e^{2r_B} R )   }}~,
\end{align}
where  $\kappa(N, R) = (-1+R)(1+2 N)$.

\begin{figure}[ptb]
	\includegraphics[
	width=3.4in
	]{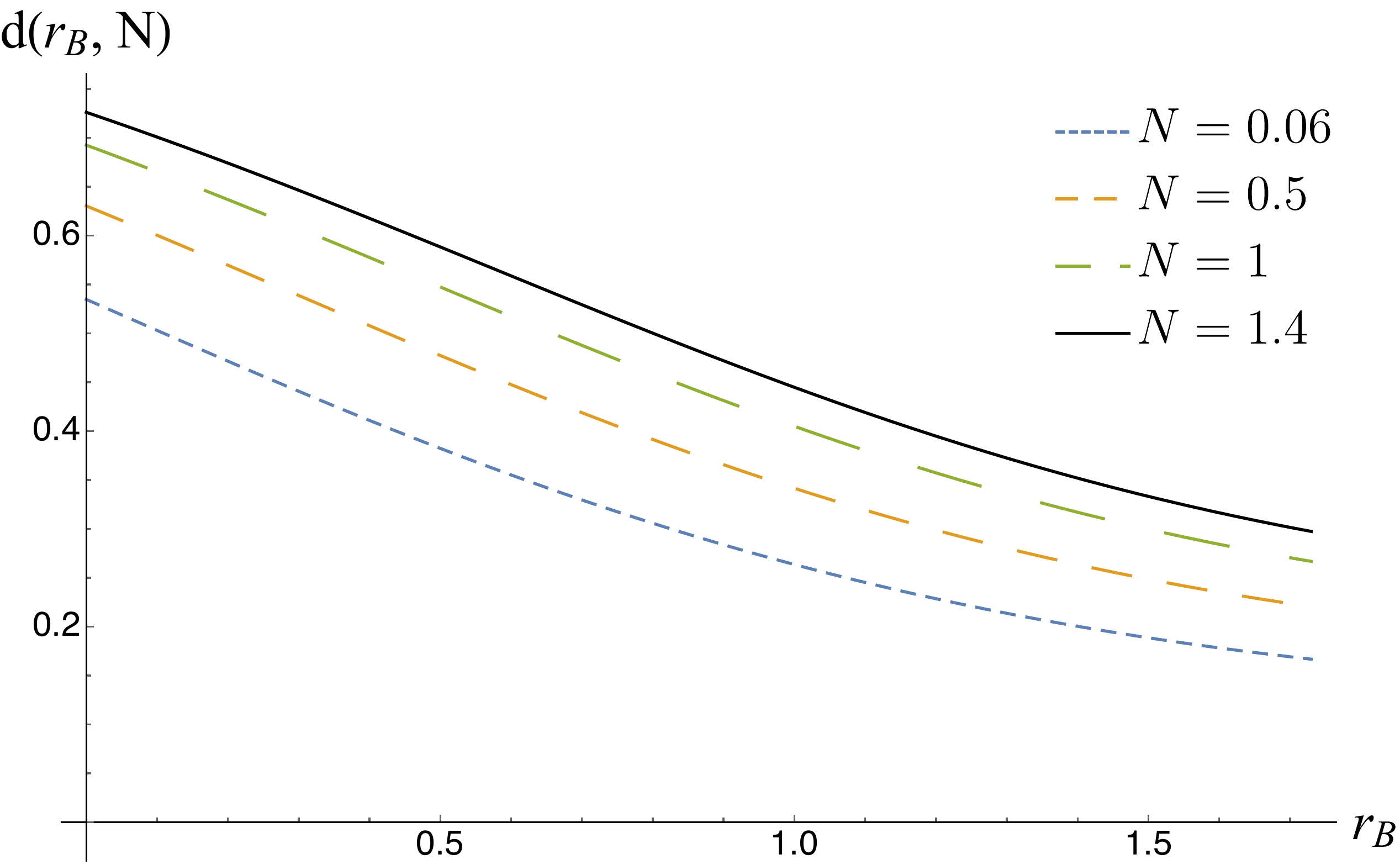}
	\caption{The figure plots the sine distance $d(r_B, N)$ in \eqref{eq:sind-SUM} between an ideal $\text{SUM}^G$ gate with the interaction gain $G=1$ and its experimental approximation $\widetilde{\text{SUM}}^{r_A, r_B, R}$ with $r_A = 1.726$ and $R=(\sqrt{5}-1)^2/4$, when the input state is a tensor product of two two-mode squeezed vacuum states with parameter $N$, as defined in \eqref{eq:tms}. In the figure, we select certain values of the mean-photon number $N$ of the channel input, with the choices indicated next to the figure. For a fixed value of $r_B$, the simulation of $\text{SUM}^G$  is more accurate for low values of the energy constraint on input states. The figure indicates that an offline squeezing strength  of 15~dB, which is what is currently experimentally achievable \cite{VMDS16}, is not sufficient to simulate an ideal $\text{SUM}^G$ gate for  $G=1$, with a high accuracy, by using the protocol from \cite{FMA05}.}
	\label{fig:energy-constraind-SUM}
\end{figure}

 We now perform numerical evaluations to see how close the experimental approximation $\widetilde{\text{SUM}}^{r_A, r_B, R}$ is to the ideal $\text{SUM}^G$ gate for a fixed input state $\psi^{\otimes 2}_{\operatorname{TMS}}(N)$. From \eqref{eq:sind-SUM}, it is evident that the sine distance between $\text{SUM}^G(\psi^{\otimes 2}_{\operatorname{TMS}}(N))$ and $\widetilde{\text{SUM}}^{r_A, r_B, R}(\psi^{\otimes 2}_{\operatorname{TMS}}(N))$ is symmetric in $r_A$ and $r_B$. Therefore, we fix $r_A = 1.726$, which corresponds to the currently experimentally  achievable maximum squeezing ($\approx 15 \text{dB}$) \cite{VMDS16}. We also fix the gain parameter  $G=1$, which implies that $R= (\sqrt{5}-1)^2/4$.

 Let $d(r_B, N)$ denote the sine distance between $\text{SUM}^G(\psi^{\otimes 2}_{\operatorname{TMS}}(N))$ and $\widetilde{\text{SUM}}^{r_A, r_B, R}(\psi^{\otimes 2}_{\operatorname{TMS}}(N))$ for $R=(\sqrt{5}-1)^2/4$ and $r_A = 1.726$:
\begin{align}\label{eq:sind-SUM}
&d(r_B, N) \nonumber \\
& =\sqrt{1-  \frac{2 R e^{r_A + r_B}}{\sqrt{ (\kappa(N, R) - 2 e^{2r_A} R   )( \kappa(N, R) - 2e^{2r_B} R )   }}} ,
\end{align}
where we used \eqref{eq:fid-SUM}.

In Figure \ref{fig:energy-constraind-SUM}, we plot $d(r_B, N)$ in \eqref{eq:sind-SUM} versus the offline squeezing strength $r_B$ for certain values of the input mean photon number $N$. Similar to the results in Sections \ref{sec:displacement-approx} and \ref{sec:sms}, we find that the simulation of $\text{SUM}^G$ is more accurate for low values of the energy constraint on  input states. Moreover, even with a low mean photon number $N=0.06$ of the input states and with the currently experimentally achievable offline squeezing strength of 15~dB, only approximately $83\%$ accuracy in simulating the ideal $\text{SUM}^{G}$ gate for $G=1$ can be achieved. 

 It is an open question to establish analytical bounds to quantify the performance of these experimental approximations of a SUM gate with respect to an energy-constrained distance measure.

\section{Approximation of one- and two-mode Gaussian unitaries}

In this section, we show that a sequence of one-mode Gaussian channels does not converge uniformly to a one-mode Gaussian unitary. The same is true for the two-mode  case. However, convergence occurs in the strong sense.

We begin by defining a set of matrices that characterizes all $n$-mode Gaussian channels. Let $\G^{X, Y}$ denote an $n$-mode Gaussian channel, which is completely characterized by two $2n \times 2n$ real matrices $X$ and $Y$. For $\G^{X, Y}$ to be a physical channel, the $X$ and $Y$ matrices must be such that 
\begin{align}\label{eq:channel-cond}
Y + i \Omega \geq i X \Omega X^T ,  \qquad Y = Y^T. 
\end{align}
Let $S_n$ denote a set of a pair of matrices $X$ and $Y$ that satisfy \eqref{eq:channel-cond}, i.e., 
\begin{align}\label{eq:set-Gauss-ch-matrices}
S_n = \{(X,Y): Y + i \Omega \geq i X \Omega X^T , Y = Y^T \}, 
\end{align}
where $n$ in the subscript of $S_n$ indicates that the set $S_n$ consists of a pair of $2n \times 2n$ real matrices. 

Let $\U_{A\to B}$ denote a single-mode Gaussian unitary transformation. Suppose that an experimental approximation of $\U_{A \to B}$ is a single-mode Gaussian channel $\widetilde{\U}_{A \to B}^{X, Y}$, which is completely characterized by two $2\times 2$ real matrices $X$ and $Y$. We now show that the sequence $\{\widetilde{\U}^{X, Y}\}_{(X, Y) \in S_1}$ does not converge uniformly to $\U$, where $S_1$ is given by \eqref{eq:set-Gauss-ch-matrices} for $n=1$. 
Let $\psi_{RA}(\bar{n})$ denote a two-mode squeezed vacuum state with parameter $\bar{n}$, as defined in \eqref{eq:tms}. 
Let $\G^{\tilde{X}, \tilde{Y}}_{A\to B} = \U^{-1}\circ\widetilde{\U}_{A \to B}^{X, Y}$ denote the overall Gaussian channel. Let 
\begin{align}
\tilde{X} = 
\begin{bmatrix}
 x_{11} & x_{12} \\ 
 x_{21} & x_{22}  
 \end{bmatrix}~,
 \tilde{Y} = 
 \begin{bmatrix}
y_{11} & y_{12}\\
y_{12} & y_{22} 
 \end{bmatrix}. 
\end{align} 
Then from the unitary invariance of fidelity, we find that
 \begin{multline}
 F(\U_{A \to B}(\psi_{RA}(\bar{n})),  \widetilde{\U}_{A \to B}^{X, Y}(\psi_{RA}(\bar{n}))) \\
  =F(\psi_{RA}(\bar{n}),\G^{\tilde{X}, \tilde{Y}}_{A\to B}(\psi_{RA}(\bar{n}))),
 \end{multline}
 where it is implicit that an identity channel acts on the reference system $R$. 
 
 By expanding $F(\U_{A \to B}(\psi_{RA}(\bar{n})), \widetilde{\U}_{A \to B}^{X, Y}(\psi_{RA}(\bar{n})))$ about $\bar{n} = \infty$, we find that \cite{Mathematica}
 \begin{align}
 &F(\U_{A \to B}(\psi_{RA}(\bar{n})),\widetilde{\U}_{A \to B}^{X, Y}(\psi_{RA}(\bar{n}))) \nonumber \\
 & = \frac{1}{(-1+x_{11} + x_{12}x_{21}+x_{22}-x_{11}x_{22})\bar{n}^2 } + O(1/\bar{n})^3. 
 \end{align}
Therefore, 
\begin{align}\label{eq:fidelity-zero-arbit-GU}
\lim_{\bar{n} \to \infty} F(\U_{A \to B}(\psi_{RA}(\bar{n})),\widetilde{\U}_{A \to B}^{X, Y}(\psi_{RA}(\bar{n}))) = 0. 
\end{align}
Using \eqref{eq:powers} and \eqref{eq:fidelity-zero-arbit-GU}, we find that 
\begin{align}
\lim_{\bar{n}\to\infty}\Vert\U_{A \to B} - \widetilde{\U}_{A \to B}^{X, Y} \Vert_{\diamond} = 2,
\end{align}
which implies that the sequence $\{\widetilde{\U}^{X, Y}\}_{(X, Y) \in S_1}$ does not converge uniformly to the ideal single-mode Gaussian unitary transformation $\U_{A \to B}$. 

Similarly, it can be shown that a sequence of two-mode Gaussian channels does not converge uniformly to an ideal two-mode Gaussian unitary transformation \cite{Mathematica}.

On the other hand, as a consequence of  \cite[Proposition~1]{M18}, the convergence holds in the strong sense for both the one- and two-mode case, and in fact for the general $n$-mode case.

\section{Conclusion}

In this paper, we studied different performance metrics to analyze how well an ideal displacement operator, an ideal single-mode squeezer, and an ideal SUM gate can be simulated experimentally. In particular, we proved that none of these experimental approximations converge uniformly to the ideal Gaussian processes. Rather, convergence occurs in the strong sense. 

We also discussed the notion of uniform convergence on the set of density operators whose marginals on the channel input have bounded energy, which is the most relevant from an experimental perspective, given that experiments are generally energy sensitive. In particular, we reduced the problem of distinguishing an ideal displacement operator from its experimental approximation to the problem of distinguishing a pure-loss channel from an ideal channel. We provided an analytic expression for the energy-constrained sine distance between an ideal displacement unitary and its approximation in terms of experimental parameters, by using the result of \cite{N18}. Moreover, we established two different lower bounds on the Shirokov--Winter energy-constrained diamond distance between an ideal displacement operator and its experimental approximation for low values of the energy constraint on input states. These bounds could be used to determine the requirements needed to implement a displacement operator to any desired accuracy. The displacement operator is ubiquitous in quantum optics and plays a critical role in CV quantum teleportation, CV quantum error correction and quantum computation, and quantum metrology. Therefore, quantification of the accuracy in simulating a displacement operator is important for several practical applications. 

We then introduced two different methods to model the noise or loss in implementing both beamsplitters and phase rotations. For these models, we established analytical bounds on the Shirokov--Winter energy-constrained diamond distance between these ideal gates and their experimental approximations. These bounds are relevant for characterizing the performance of any experiment consisting of beamsplitters  and phase rotations. 

Similarly, we discussed the notion of uniform convergence on the set of density operators whose marginals on the channel input have bounded energy for experimental approximations of both the single-mode squeezing unitary and the SUM gate. We considered several experimentally relevant input quantum states and studied how close these experimental approximations are to the ideal quantum processes. It is an interesting open question to determine the optimal value of energy-constrained distance measures and the corresponding optimal state to completely characterize these experimental approximations of the ideal quantum processes. 

In this paper, homodyne measurements involved in simulating a single-mode squeezer and a SUM gate were considered ideal. We expect that the well-known experimental approximation of homodyne detection converges strongly to  ideal homodyne detection, based on the calculation of \cite[Appendix~K]{WPS01}, and we also expect that the experimental approximation will not converge uniformly.  However, it is an open question to determine the optimal value of energy-constrained distance measures and corresponding optimal states when homodyne measurements involved in these simulations are not ideal. Another interesting direction is to use these results to study the error propagation in an experiment based on quantum optical elements.  We leave this for future work.

\begin{acknowledgments}
We thank Lior Cohen, Jonathan P. Dowling, A.~R.~P.~Rau, and  Barry Sanders for
discussions related to the topic of this paper. We are indebted to an anonymous referee for the suggestion to include an analysis for beamsplitters and phase rotations. MMW acknowledges support from the Office
of Naval Research and the National Science Foundation. KS acknowledges
support from the National Science
Foundation under Grant No.~1714215.
\end{acknowledgments}

\bibliographystyle{unsrt}
\bibliography{Ref}

\begin{thebibliography}{10}
\bibitem{Mathematica}
Mathematica and {M}atlab files are available in the source files of our arxiv
  post.
\end{thebibliography}

\appendix


\section{Preliminaries}
\label{sec:review}

We begin by reviewing some definitions and prior results relevant for the rest of the appendices. We point readers to \cite{H12,AS17} for background.

\bigskip 

\noindent \textbf{Gaussian states and channels} Let $\H$ denote an infinite-dimensional, separable Hilbert space. Let $\T(\H)$ denote the set of trace-class operators, i.e., all operators $M$ with finite trace norm: $\left\Vert M \right\Vert_1 \equiv \tr(\sqrt{M^{\dagger}M}) < \infty$. Let $\D(\H)$ denote the set of density operators acting on $\H$, i.e., those that are positive semi-definite with unit trace. The continuous-variable system of interest in this work is $n$ quantized modes of the electromagnetic field. Any physical state of $n$ bosonic modes can be described by density operators acting on a tensor-product Hilbert space $\H^{\otimes n} = \otimes_{i=1}^n \H_i$, where $\H_i$ is the Hilbert space corresponding to the $i$th mode. Let $\hat{x}_i$ and $ \hat{p}_i$ denote the respective position- and momentum-quadrature operators of the $i$th mode. Let $\hat{r} \equiv (\hat{x}_1, \hat{p}_1, \dots, \hat{x}_n, \hat{p}_n)^T$. Then the following commutation relation holds:
\begin{align}
[\hat{r}, \hat{r}^T] = i \Omega~, 
\end{align}
where $\Omega = \bigoplus_{i=1}^{n} \Omega_0$, and $\Omega_0 = \begin{bmatrix}
0 & 1 \\
-1 & 0
\end{bmatrix}$. Furthermore, we take the annihilation operator for the $i$th mode as $\hat{a}_i = (\hat{x}_i + i \hat{p}_i)/\sqrt{2}$. 

For $r \in \mathbb{R}^{2n}$, we define the unitary displacement operator $D(r) \equiv \exp(i r^T \Omega \hat{r})$. Moreover, the set $\{D(r)\}_r$ forms an orthogonal complete set on the space of operators acting on the Hilbert space $L^2(\mathbb{R}^n)$ of square integrable functions, in the sense that $\tr(D(r)D(-r^{\prime}) ) = (2\pi)^{n}\delta^{2n}(r-r^{\prime})$. Therefore, any quantum state $\rho \in \D(\H)$ can be represented as follows:
\begin{align}
\rho = \frac{1}{(2\pi)^n}\int d^{2n} r~\chi_{\rho} (r) D(r)~,
\end{align}
where $\chi_{\rho}(r) \equiv \tr(D(-r) \rho )$ is the Wigner characteristic function of the state $\rho$. Moreover, a quantum state $\rho$ is Gaussian if its characteristic function has the following form:
\begin{align}
\chi_{\rho}(r) = \exp\!\left(-\frac{1}{4}r^T \Omega^T V_{\rho} \Omega r + i r^T \Omega^T \mu_{\rho}\right),
\end{align}
where $\mu_{\rho} \equiv \langle \hat{r} \rangle_{\rho}$ is the mean vector of $\rho$ and $V_{\rho} \equiv \langle \{(\hat{r} - \mu_{\rho}), (\hat{r} - \mu_{\rho})^T \}\rangle_{\rho}$ is the covariance matrix.

Quantum channels that take any Gaussian input state to another Gaussian state are called quantum Gaussian channels. Let $\rho$ be an input state with the characteristic function $\chi_{\rho}(r)$. Then under the action of a Gaussian channel~$\N$ from $n$ modes to $m$ modes, $\chi_{\rho}(r)$ transforms as follows:
\begin{multline}
\label{eq:Wigner-charac-transformation}
\chi_{\rho}(r) \to \chi_{\N(\rho)}(r) \\ = \chi_{\rho}\!\left(\Omega^TX^T\Omega r\right)\exp\!\left(-\frac{1}{4}r^T \Omega^T Y \Omega r + i r^T \Omega^T d\right)~,
\end{multline}
where $X$ is a real $2m \times 2n$ matrix, $Y$ is a real $2m \times 2m$ positive semi-definite symmetric matrix, and $d\in \mathbb{R}^{2m}$, such that they satisfy the following condition for $\N$ to be a physical channel:
\begin{align}\label{eq:cptp-Gaussian-condition}
Y + i \Omega - i X \Omega X^T  \geq 0~. 
\end{align}
Furthermore, since a Gaussian state $\rho$ can be completely characterized by its mean vector $\mu_{\rho}$ and covariance matrix~$V_{\rho}$, the action of a Gaussian channel on  $\rho$ can be described as follows
\begin{align}
&\mu_{\rho} \to X \mu_{\rho} + d~, \nonumber \\
& V_{\rho} \to X V_{\rho} X^T + Y~. 
\end{align}

Let $Y=0$. Then from \eqref{eq:cptp-Gaussian-condition} we get $X \Omega X^T = \Omega$, which further implies that $X$ is an element of the real symplectic group $Sp(2n, \mathbb{R})$. The symplectic group is a set of transformations that preserve the anti-symmetric form $\Omega$ when acting by congruence, i.e., $S \Omega S^T = \Omega,~ \forall S\in Sp(2n, \mathbb{R})$. Therefore, the group of Gaussian unitaries is identified with $Sp(2n, \mathbb{R})$. 

Let $\rho$ be an $n$-mode Gaussian quantum state with the mean vector $\mu_{\rho}$ and the covariance matrix $V_{\rho}$. The Wigner function of $\rho$ is given by
\begin{align}
W(r) = \frac{2^n}{\pi^n\sqrt{\text{Det}(V_{\rho})}} \exp\left[-(r-\mu_{\rho})^T V_{\rho}^{-1}(r-\mu_{\rho}) \right]. 
\end{align}

Gaussian state transformations can also be described in the phase-space formalism. In particular,  the action of a symplectic transformation $S$ on a Gaussian state is given by
\begin{align}
W(r) \to W(S^{-1}r)~.
\end{align}
Moreover, the Wigner function of a Gaussian input state transforms under a linear displacement $D(-\bar{r})\rho D(\bar{r})$ as
\begin{align}
W(r) \to W_G(r-\bar{r})~.
\end{align}

A coherent state $\ket{\alpha}$ is an eigenvector of the annihilation operator $\hat{a}$ with eigenvalue $\alpha$, i.e., $\hat{a} \ket{\alpha} = \alpha \ket{\alpha}$.  The coherent state $\ket{\alpha}$ can also be represented as $ \ket{\alpha} = D(\alpha) \ket{0}$. Moreover, the overlap between two coherent states $\ket{\alpha}$ and $\ket{\beta}$ is given by
\begin{align} \label{eq:overlap-coherent-states}
\langle \beta \vert \alpha \rangle = \exp\left(-\frac{1}{2}\vert \alpha - \beta \vert^2\right) \exp\left[\frac{1}{2}(\alpha \beta^{*} - \alpha^*\beta)\right]~.
\end{align}

A single-mode thermal state with mean photon number $\bar{n} = 1/(e^{\beta \omega} -1)$ has the following representation in the photon number basis:
\begin{equation}\label{eq:thermalstate}
\theta(\bar{n}) \equiv \frac{1}{1+\bar{n}}\sum_{n=0}^{\infty} \left(\frac{\bar{n}}{\bar{n}+1}\right)^n \ket{n}\bra{n}~.
\end{equation}

In our paper, we employ the two-mode squeezed vacuum state with parameter $\bar{n}$, which is equivalent to a purification of the thermal state in \eqref{eq:thermalstate} and is defined as
\begin{equation}
\ket{\psi_{\operatorname{TMS}}(\bar{n})} \equiv \frac{1}{\sqrt{\bar{n}+1}}\sum_{n=0}^{\infty}\sqrt{\left(\frac{\bar{n}}{\bar{n}+1}\right)^n} \ket{n}_R\ket{n}_A~.
\end{equation}

\textbf{Quantum pure-loss channel.} A quantum pure-loss channel is a Gaussian channel that can be characterized by a beamsplitter of transmissivity $\eta\in (0,1)$, coupling the signal input state with the vacuum state, and followed by a partial trace over the environment. In the Heisenberg picture, the beamsplitter transformation is given by the following Bogoliubov transformation:
\begin{align}\label{eq:beam-splitter-transformation1}
&\hat{b} = \sqrt{\eta} \hat{a} - \sqrt{1-\eta} \hat{e},\\ 
&\hat{e}' = \sqrt{1-\eta}\hat{a} + \sqrt{\eta} \hat{e},
\end{align}
where $\hat{a}, \hat{b}$, $\hat{e}$, and $\hat{e}'$  are the annihilation operators representing the sender's input mode, the receiver's output mode, an environmental input mode, and an environmental output mode of the channel, respectively.

 \textbf{Topologies of convergence}.
Uniform and strong convergence in the context of infinite-dimensional quantum channels were studied in \cite{SH07}. A connection between the notion of strong convergence and the notion of uniform convergence over energy-bounded states was established in \cite{S18}. Later, these different topologies of convergence were studied in the context of linear bosonic channels and Gaussian dilatable channels in \cite{LSW18}. Furthermore, topologies of convergence in the context of teleportation simulation of physically relevant phase-insensitive bosonic Gaussian channels have been investigated in \cite{M18}.

\section{Convergence of the experimental implementation of a displacement operator}
In this section, we begin by providing a detailed proof of the convergence of the experimental implementation of a displacement operator from \cite{P96} to the ideal displacement operator. We then provide a proof for \eqref{eq:trace-distance-for-displacement}.

We begin by showing that the channel corresponding to the experimental implementation of a displacement operator is equivalent to a pure-loss channel followed by the ideal displacement operator. Consider that
\begin{align}
& (\operatorname{Tr}_{B}\circ\mathcal{B}_{AB}^{\eta})(\rho_{A}\otimes
|\beta\rangle\langle\beta|_{B}) \notag \\
&=(\operatorname{Tr}_{B}\circ
\mathcal{B}_{AB}^{\eta}\circ\mathcal{D}_{B}^{\beta})(\rho_{A}\otimes
|0\rangle\langle0|_{B})\\
& =(\operatorname{Tr}_{B}\circ\lbrack\mathcal{D}_{A}^{\sqrt{1-\eta}%
\beta}\otimes\mathcal{D}_{B}^{\sqrt{\eta}\beta}]\circ\mathcal{B}_{AB}^{\eta
})(\rho_{A}\otimes|0\rangle\langle0|_{B})\\
& =(\operatorname{Tr}_{B}\circ\mathcal{D}_{A}^{\sqrt{1-\eta}\beta
}\circ\mathcal{B}_{AB}^{\eta})(\rho_{A}\otimes|0\rangle\langle0|_{B})\\
& =(\mathcal{D}_{A}^{\sqrt{1-\eta}\beta}\circ\operatorname{Tr}_{B}%
\circ\mathcal{B}_{AB}^{\eta})(\rho_{A}\otimes|0\rangle\langle0|_{B})\\
& =(\mathcal{D}_{A}^{\alpha}\circ\mathcal{L}_{A}^{\eta})(\rho_{A}). \label{eq:displacement-approx-equiv-pure-loss}
\end{align}
The first equality follows from the definition of a coherent state. The second equality follows from the following covariance of the beamsplitter unitary with respect to displacement operators \cite{AS17}:
\begin{equation}
\mathcal{B}_{AB}^{\eta}\circ\mathcal{D}_{B}^{\beta}=[\mathcal{D}_{A}^{\sqrt{1-\eta}\beta}\otimes\mathcal{D}_{B}^{\sqrt{\eta}\beta}]\circ
\mathcal{B}_{AB}^{\eta}~.
\end{equation}
The third equality follows from the cyclicity of partial trace. In the final equality we  defined the pure-loss channel as
$\mathcal{L}_{A}^{\eta}(\rho_{A})=(\operatorname{Tr}_{B}\circ\mathcal{B}%
_{AB}^{\eta})(\rho_{A}\otimes|0\rangle\langle0|_{B})$.

Let $\psi_{RA}$ be an arbitrary two-mode state. To compute the fidelity
between the ideal displacement operator and its experimental approximation,
consider that%
\begin{equation}
F\big(\mathcal{D}_{A}^{\alpha}(\psi_{RA}),(\mathcal{D}_{A}^{\alpha}\circ
\mathcal{L}_{A}^{\eta})(\psi_{RA})\big)=F\big(\psi_{RA},\mathcal{L}_{A}^{\eta}%
(\psi_{RA})\big),
\end{equation}
where we employed the unitary invariance of the fidelity. 

Let $\ket{\delta}$ be a coherent state. Then $\ket{\delta}$ transforms under the pure-loss channel $\mathcal{L}_{A}^{\eta}$ to $\ket{\sqrt{\eta} \delta}$. Therefore, by using \eqref{eq:overlap-coherent-states} we get
\begin{align}
F(\ket{\delta}\bra{\delta},\mathcal{L}_{A}^{\eta}(\ket{\delta}\bra{\delta})) = \exp\left[- \vert \delta \vert^2 (1-\sqrt{\eta})^2\right] ~,
\end{align}
which converges to zero as $\vert \delta \vert^2 \to \infty$, and in turn implies that the sequence  $\{\widetilde{\D}^{\eta, \frac{\alpha}{\sqrt{1-\eta}}}\}_{\eta\in [0,1]}$ does not converge uniformly to the ideal displacement channel $\D^{\alpha}$. 

We now show that convergence occurs in the strong sense. Let $\rho_A$ denote the input state.  Let $\chi_{\rho_A}(x, p)$ denote the Wigner characteristic function of the state $\rho_A$. Then from \eqref{eq:Wigner-charac-transformation}, the Wigner characteristic function of the output of an ideal displacement channel is given by
\begin{align}
\chi_{\D^{\alpha}(\rho_A)}(x, p ) = \chi_{\rho_A}(x, p) \exp{\big[i \sqrt{2} ( p\text{Re}(\alpha)-x\text{Im}(\alpha))\big]}~. 
\end{align}

We now find $X$, $Y$ matrices and the $d$ vector corresponding to the Gaussian channel $\widetilde{\D}^{\eta, \frac{\alpha}{\sqrt{1-\eta}}}$. By using \eqref{eq:displacement-approx-equiv-pure-loss}, we get $X =\text{diag}(\sqrt{\eta}, \sqrt{\eta}) $, $Y = \text{diag}(1-\eta, 1-\eta)$, and $d = (\sqrt{2}\text{Re}(\alpha), \sqrt{2}\text{Im}(\alpha))^{\text{T}}$. Let $\tilde{\rho}^{\text{out}}_A = \widetilde{\D}^{\eta, \frac{\alpha}{\sqrt{1-\eta}}}(\rho_{A})$. Then from \eqref{eq:Wigner-charac-transformation}, the Wigner characteristic function of $\tilde{\rho}^{\text{out}}_A$ is given by
\begin{multline}
\chi_{\tilde{\rho}^{\text{out}}_A} (x, p)= \chi_{\rho_A}(\sqrt{\eta}x, \sqrt{\eta}p) \times \\
\exp{\big[i \sqrt{2} (  p\text{Re}(\alpha)-x\text{Im}(\alpha) ) -(1/4)(x^2+p^2)(1-\eta)\big]}~. 
\end{multline}
Therefore, for each $\rho_A \in \D(\H_A)$, and for all $x, p \in \mathbb{R}$
\begin{align}
\lim_{\eta \to 1} \chi_{\tilde{\rho}^{\text{out}}_A} (x, p) = \chi_{\D^{\alpha}(\rho_A)}(x, p)~,
\end{align}
which implies that the sequence  $\{\widetilde{\D}^{\eta, \frac{\alpha}{\sqrt{1-\eta}}}\}_{\eta\in [0,1)}$ converges to  $\D^{\alpha}$ in the strong sense \cite[Lemma~8]{LSW18}. 

We now discuss the notion of uniform convergence on the set of density operators whose marginals on the channel input have bounded energy for experimental approximations of a tensor product of ideal displacement channels.  Let $\{\D^{\alpha_i}\}_{i=1}^K$ be a set of $K$ different displacement channels. We now approximate the tensor product of these operators by a tensor product of $\{\widetilde{\D}^{\eta_i, \beta_i}\}_{i=1}^K$,  such that $\sqrt{1-\eta_i}\beta_i = \alpha_i,~ \forall i\in\{1, \dots, K\}$. Let $H_A$ denote the Hamiltonian of the system $A$. Moreover, suppose that there is an average energy constraint on the input state to the tensor product of displacement operators, i.e., $\tr(\widetilde{H}_{A^K}\psi_{A^K} ) \leq E$, where 
\begin{equation}
\widetilde{H}_{A^K} \equiv H_A\otimes I \otimes \dots \otimes I+\dots + I\otimes \dots \otimes I \otimes H_A,
\end{equation}
and $E \in [0,\infty)$. Let $\tr(H_A \psi_{A_i}) = E_i$, where $E_i \in [0,\infty)$, $\forall i\in\{1, \dots, K\}$.

\onecolumngrid

\vspace{.2in}
Consider the following chain of inequalities:
\begin{align}
 \frac{1}{2} \left\Vert \left(\bigotimes_{i=1}^K\D^{\alpha_i}\right)(\psi_{RA^K}) -  \left(\bigotimes_{i=1}^{K}\widetilde{\D}^{\eta_i, \alpha_i/\sqrt{1-\eta_i}}\right)(\psi_{RA^K})\right\Vert_{1}  &\leq \sum_{i=1}^{K} C\big(\L^{\eta_i}_{A_i}(\psi_{RA_i}), \I_{A_i}(\psi_{RA_i})  \big)\\
& \leq \max_{\{E_i\}_i: \sum_i E_i \leq E} \sum_{i=1}^{K}\sqrt{1 - [(1- \{E_i\}) \sqrt{\eta_i}^{\lfloor E_i\rfloor} + \{E_i\}\sqrt{\eta_i}^{\lceil E_i\rceil} ]^2}.
\end{align}
The first inequality follows from \eqref{eq:powers} and \cite[Proposition~1]{M18}.  The last inequality  follows from the recent result of \cite{N18}, and due to the maximization over a set of energy values satisfying the input energy constraint. Since the chain of inequalities is true for all input states satisfying the input energy constraint, the following holds
\begin{equation}
\frac{1}{2} \left\Vert \bigotimes_{i=1}^K\D^{\alpha_i} -  \bigotimes_{i=1}^{K}\widetilde{\D}^{\eta_i, \alpha_i/\sqrt{1-\eta_i}}\right\Vert_{\diamond E} 
\leq  \max_{\{E_i\}_i: \sum_i E_i \leq E} \sum_{i=1}^{K}\sqrt{1 - [(1- \{E_i\}) \sqrt{\eta_i}^{\lfloor E_i\rfloor} + \{E_i\}\sqrt{\eta_i}^{\lceil E_i\rceil} ]^2}~. 
\end{equation}
Therefore,  $\{\bigotimes_{i=1}^{K}\widetilde{\D}^{\eta_i, \alpha_i/\sqrt{1-\eta_i}}\}_{\eta_1, \dots, \eta_K \in [0,1)}$ converges uniformly to $\bigotimes_{i=1}^K\D^{\alpha_i}$ on the set of density operators whose marginals on the channel input have bounded energy.

\vspace{.1in}
We now provide a proof for \eqref{eq:trace-distance-for-displacement}. Let $t = \sqrt{\eta}$ and $r = \sqrt{1-\eta}$. Then the action of a pure-loss channel with transmissivity $\eta$ on  $\psi_{RA}$ in \eqref{eq:opt-state}
is given by 
\begin{align}
(\I_R \otimes \L_A^{\eta})(\psi_{RA})  = \tr_B((\I_R \otimes \B^{\eta}_{AB} )(\psi_{RA} \otimes \ket{0}\bra{0}_B)) \equiv \rho_{RA}~. 
\end{align}

Consider the following unitary evolution of the pure state $\ket{\psi}_{RA}\otimes \ket{0}_B$: 
\begin{align}
&(I_R \otimes B^{\eta}_{AB} )(\ket{\psi}_{RA} \otimes \ket{0}_B) \nonumber\\
&= B^{\eta}_{AB}( \sqrt{\frac{1-\{E\}}{\lfloor E\rfloor !}} (\hat{a}_{\text{in}}^{\dagger})^{\lfloor E\rfloor}    \ket{0}_A \ket{\tau}_R \ket{0}_B + \sqrt{ \frac{\{E\}}{\lceil E \rceil !}} (\hat{a}_{\text{in}}^{\dagger})^{\lceil E\rceil}     \ket{0}_A \ket{\tau^{\perp}}_R \ket{0}_B  )\\
&=  \sqrt{\frac{1-\{E\}}{\lfloor E\rfloor !}} (t \hat{a}_{\text{out}}^{\dagger}+ r \hat{b}_{\text{out}}^{\dagger})^{\lfloor E\rfloor}    \ket{0}_A \ket{\tau}_R \ket{0}_B + \sqrt{ \frac{\{E\}}{\lceil E \rceil !}} (t \hat{a}_{\text{out}}^{\dagger}+ r \hat{b}_{\text{out}}^{\dagger})^{\lceil E\rceil}     \ket{0}_A \ket{\tau^{\perp}}_R \ket{0}_B~ \label{eq:pure-state-after-BS}\\
& =\sqrt{1-\{E\}} \sum_{k_1} \sqrt{\binom{\lfloor E \rfloor}{k_1}} t^{k_1} r^{ \lfloor E \rfloor - k_1 } \ket{k_1}_A \ket{\tau}_R \ket{\lfloor E \rfloor -k_1 }_B + \sqrt{ \{ E  \}   } \sum_{k_2} \sqrt{\binom{\lceil E \rceil}{k_2}} t^{k_2} r^{ \lceil E \rceil - k_2 } \ket{k_2}_A \ket{\tau^{\perp}}_R \ket{\lceil E \rceil -k_2}_B~. 
\end{align}
The second equality follows from the beamsplitter transformation in the Heisenberg picture \eqref{eq:beam-splitter-transformation1}. Then the density operator after tracing out $B$  in \eqref{eq:pure-state-after-BS} is given by 
\begin{align}
&\rho_{RA}  = (1- \{ E  \}) \sum_{k_1} \binom{\lfloor E \rfloor}{k_1} t^{2 k_1} r^{2( \lfloor E \rfloor - k_1 )} \ket{k_1}\bra{k_1}_A \ket{\tau}\bra{\tau}_R + \{ E\} \sum_{k_2} \binom{\lceil E \rceil}{k_2} t^{2k_2} r^{2(\lceil E\rceil -k_2  )} \ket{k_2}\bra{k_2}_A \ket{\tau^{\perp}}\bra{\tau^{\perp}}_R \nonumber \\
&+ \sqrt{(1-\{E\})\{E\}} \sum_{k_1} \sqrt{\binom{\lfloor E \rfloor}{k_1} \binom{\lceil E \rceil}{k_1 + 1} } t^{ (2k_1+1)} r^{ \lfloor E  \rfloor + \lceil E \rceil - 2k_1 -1 } \Bigg(\ket{k_1}\bra{k_1+1}_A \ket{\tau}\bra{\tau^{\perp} }_R  +  \ket{k_1+1}\bra{k_1}_A \ket{\tau^{\perp}}\bra{\tau}_R \Bigg)\nonumber~. \label{eq:noisy-state}\\
\end{align}

On the other hand, the density operator of the state in \eqref{eq:opt-state} is given by
\begin{multline}
\psi_{RA} = (1 - \{ E  \}) \ket{\lfloor E \rfloor}\bra{\lfloor E \rfloor}_A \ket{\tau}\bra{\tau}_R + \{E\} \ket{\lceil E \rceil}\bra{\lceil E \rceil}_A \ket{\tau^{\perp}}\bra{\tau^{\perp}}_R  \\
 + \sqrt{  (1 - \{ E  \})\{E\}   } \bigg( \ket{\lfloor E \rfloor}\bra{\lceil E \rceil}_A \ket{\tau}\bra{\tau^{\perp}}_R + \ket{\lceil E \rceil}\bra{\lfloor E \rfloor}\ket{\tau^{\perp}}\bra{\tau}_R       \bigg)~. \label{eq:pure-state}
\end{multline}

Let $\varrho_{RA}$ denote the operator corresponding to the difference of  $\rho_{RA}$ in \eqref{eq:noisy-state} and $\psi_{RA}$ in \eqref{eq:pure-state}. 
\begin{align}
\varrho_{RA} &= \rho_{RA} - \psi_{RA}  \\
 &= (1- \{ E  \}) \Bigg[ \sum_{k_1} \binom{\lfloor E \rfloor}{k_1} t^{2 k_1} r^{2( \lfloor E \rfloor - k_1 )} \ket{k_1}\bra{k_1}_A - \ket{\lfloor E \rfloor}\bra{\lfloor E \rfloor}_A \Bigg] \ket{\tau}\bra{\tau}_R \nonumber \\
& +  \{ E\} \Bigg[ \sum_{k_2} \binom{\lceil E \rceil}{k_2} t^{2k_2} r^{2(\lceil E\rceil -k_2  )} \ket{k_2}\bra{k_2}_A -  \ket{\lceil E \rceil}\bra{\lceil E \rceil}_A\Bigg] \ket{\tau^{\perp}}\bra{\tau^{\perp}}_R \nonumber \\
& + \sqrt{(1-\{E\})\{E\}}\Bigg[ \sum_{k_1} \sqrt{\binom{\lfloor E \rfloor}{k_1} \binom{\lceil E \rceil}{k_1 + 1} } t^{ (2k_1+1)} r^{ \lfloor E  \rfloor + \lceil E \rceil - 2k_1 -1 }  \Bigg(\ket{k_1}\bra{k_1+1}_A \ket{\tau}\bra{\tau^{\perp} }_R  +  \ket{k_1+1}\bra{k_1}_A \ket{\tau^{\perp}}\bra{\tau}_R \Bigg) \nonumber \\
& - \ket{\lfloor E \rfloor}\bra{\lceil E \rceil}_A \ket{\tau}\bra{\tau^{\perp}}_R - \ket{\lceil E \rceil}\bra{\lfloor E \rfloor}\ket{\tau^{\perp}}\bra{\tau}_R \Bigg] ~. 
\end{align}

We now find $\Vert \varrho_{RA}\Vert_{1}$ for a simple case. Let $0<E<1$. Then $\lfloor E \rfloor = 0$ and $\lceil E\rceil = 1$. Let $\ket{\tau} = \ket{0}$ and $\ket{\tau^{\perp}} = \ket{1}$. Therefore, the operator $\varrho_{RA}$ is given by

\begin{align}
\varrho_{RA} =\{ E \}\Bigg[  r^2 \ket{0}\bra{0}_A + t^2 \ket{1}\bra{1}_A - \ket{1}\bra{1}_A     \Bigg] \ket{1}\bra{1}_R   + \sqrt{(1-  \{  E \}  ) \{E\}   } \Bigg[ (t -1)\bigg(  \ket{0}\bra{1}_A \otimes \ket{0}\bra{1}_R + \ket{1}\bra{0}_A \otimes \ket{1}\bra{0}_R     \bigg)       \Bigg]~.
\end{align}

After expressing $r$ and $t$ in terms of $\eta$, the matrix representation of $\varrho_{RA}$ is as follows
\begin{equation} \varrho_{RA}
=
\begin{bmatrix}
0 & 0 & 0 & - \sqrt{(1-  \{  E \}  ) \{E\}   }(1-\sqrt{\eta})\\
0 & \{ E \} (1-\eta) & 0 & 0\\
0 & 0 & 0 & 0\\
- \sqrt{(1-  \{  E \}  ) \{E\}   }(1-\sqrt{\eta}) & 0 & 0 & -\{ E \}(1-\eta)
\end{bmatrix}
\end{equation}

Then $\Vert \varrho_{RA}\Vert_{1}$ is given by \cite{Mathematica}
\begin{align}
\frac{1}{2}\left\Vert  \varrho_{RA} \right\Vert_1  &= \frac{1}{2}\Bigg[ \{E\}(1-\eta) +\frac{1}{2} \left\vert  (\sqrt{\eta} -1 )\Bigg(\{ E  \} (1+ \sqrt{\eta}) -\sqrt{ \{E\}\bigg[4 + \{E\}(    -3 + 2\sqrt{\eta}+ \eta) \bigg]  }  \Bigg)  \right\vert \nonumber \\ 
&\qquad +\frac{1}{2} \left\vert  (\sqrt{\eta} -1 )\Bigg(\{ E  \} (1+ \sqrt{\eta}) +\sqrt{ \{E\}\bigg[4 + \{E\}(    -3 + 2\sqrt{\eta}+ \eta) \bigg]  }  \Bigg)  \right\vert\Bigg]. 
\end{align}
Since 
\begin{align}
\{E\}(1+\sqrt{\eta}) \leq \sqrt{ \{E\}\bigg[4 + \{E\}(    -3 + 2\sqrt{\eta}+ \eta) \bigg]  }  \Bigg) 
\end{align}
for all $0< E<1$ and $\eta\geq 0$, we get that 
\begin{align}
\frac{1}{2}\Vert \varrho_{RA}\Vert_1 = \frac{1}{2} \Bigg[   \{E\}(1-\eta ) + (1-\sqrt{\eta}) \sqrt{\{E\}  \bigg[4 + \{E\}(    -3 + 2\sqrt{\eta}+ \eta) \bigg]  } \Bigg]~.
\end{align}

\section{Convergence of the experimental implementation of a beamsplitter}\label{app:bs}
In this section, we provide a proof for \eqref{eq:strong-conv-bs1}. Let $\chi_{\rho_{A_1A_2}}(x_1, p_1, x_2, p_2)$ denote the Wigner characteristic function for the input state $\rho_{A_1A_2}$. 
Let $\tilde{\rho}^{\text{out}}_{A_1A_2}(\eta,\phi, \eta')$ denote the state after the action of $\widetilde{\B}^{\eta, \phi, \eta'}$ on $\rho_{A_1A_2}$:
\begin{align}
\tilde{\rho}^{\text{out}}_{A_1A_2}(\eta,\phi, \eta') = \widetilde{\B}^{\eta, \phi, \eta'}(\rho_{A_1A_2})~.
\end{align}
We now find the terms involved in \eqref{eq:Wigner-charac-transformation} for both $\B^{\eta, \phi}(\rho_{A_1A_2})$ and $\rho^{\eta, \phi, \eta'}_{A_1A_2}$.  $X_{\widetilde{\B}^{\eta, \phi, \eta'}}$ matrix corresponding to the operation $\widetilde{\B}^{\eta, \phi, \eta'}$ is given by 
\begin{equation}
X_{\widetilde{\B}^{\eta, \phi, \eta'}} = X_{\B^{\eta, \phi}} \cdot \sqrt{\eta'}I_4,
\end{equation}
where $I_4$ is a four-dimensional identity matrix. Moreover, $Y_{\widetilde{\B}^{\eta, \phi, \eta'}}$ matrix is given by 
\begin{equation}
Y_{\widetilde{\B}^{\eta, \phi, \eta'}} = (1-\eta')I_4.
\end{equation}
Let $r = (x_1, p_1, x_2, p_2)^{\text{T}}$. Then, for each $\rho_{A_1A_2}\in \D(\H_{A_1}\otimes \H_{A_2})$, and for all $x_1, p_1, x_2, p_2 \in \mathbb{R}$
\begin{align} 
\lim_{\eta' \to 1} \chi_{\tilde{\rho}^{\text{out}}_{A_1A_2}(\eta,\phi, \eta')} (r) 
&= \lim_{\eta' \to 1} \chi_{\rho_{A_1A_2}}(\sqrt{\eta'} \Omega^{\text{T}}X_{\B^{\eta, \phi}}^{\text{T}}\Omega r ) \exp(-\frac{1}{4}(1-\eta')r^{\text{T}}r)\\
&= \chi_{\rho_{A_1A_2}}( \Omega^{\text{T}}X_{\B^{\eta, \phi}}^{\text{T}}\Omega r ) \\
& = \chi_{\B^{\eta, \phi}(\rho_{A_1A_2})}(r).
\end{align} 
We have thus shown that the sequence of characteristic functions $\chi_{\tilde{\rho}^{\text{out}}_{A_1A_2}(\eta,\phi, \eta')}$ converges pointwise to $\chi_{\B^{\eta, \phi}(\rho_{A_1A_2})}$, which implies
by \cite[Lemma~8]{LSW18}
that the sequence $\{\widetilde{\B}^{\eta,\phi, \eta'}\}_{\eta' \in[0,1)}$  converges to $\B^{\eta, \phi}$ in the strong sense.

\end{document}